\documentclass[lettersize,journal]{IEEEtran}
\usepackage{amsmath,amsfonts}
\usepackage{algorithmic}
\usepackage{algorithm}
\usepackage{array}
\usepackage{textcomp}
\usepackage{stfloats}
\usepackage{url}
\usepackage{verbatim}
\usepackage{graphicx}
\usepackage{cite}
\usepackage{amsthm,amsmath,amssymb}
\allowdisplaybreaks
\usepackage{mathrsfs}
\usepackage{enumitem}
\usepackage{subcaption} 
\usepackage{times}
\usepackage{svg}
\usepackage{multirow}
\usepackage{amsthm}            
\usepackage{longtable}
\usepackage{booktabs}
\usepackage{dsfont}
\usepackage{tabularx}
\usepackage{amsfonts,bm,mathtools}
\usepackage{dsfont}

\usepackage{xcolor}

\usepackage{threeparttable}
\usepackage{colortbl}

\theoremstyle{plain}

\newtheorem{corollary}{Corollary}

\newtheorem{proposition}{Proposition}

\allowdisplaybreaks

\begin{document}

\title{AP Association for RHS-Enabled Cell-Free Uplink MIMO in Industrial Indoor UAV Networks}

\author{Liangshun Wu, Wen Chen*, Zhendong Li, Qiong Wu and Ying Wang
\thanks{* Corresponding author: Wen Chen.}
\thanks{Liangshun Wu  and  Wen Chen are with the Department of Electronic Engineering, Shanghai Jiao Tong University, Shanghai, China (e-mail: wuliangshun@sjtu.edu.cn; wenchen@sjtu.edu.cn).  Zhendong Li is with the School of Information and Communication Engineering, Xi’an Jiaotong University, Xi’an, China (e-mail: lizhendong@xjtu.edu.cn). Qiong Wu is with the School of Internet of Things Engineering, Jiangnan University, Wuxi, China (e-mail: qiongwu@jiangnan.edu.cn). Ying Wang is with the State Key Laboratory of Networking and Switching Technology, Beijing University of Posts and Telecommunications, Beijing, China (e-mail: wangying@bupt.edu.cn).}
}


\maketitle

\begin{abstract}
Indoor industrial UAV uplink networks face serious blockage and shadowing from shelves, metal equipment, and production facilities. UAVs are also often clustered and fly along similar straight inspection routes at fixed heights. These features make traditional small-cell deployment less suitable, especially when high reliability, continuous coverage, and good service for weak UAVs are required. Cell-free networks can improve robustness through distributed access points (APs) and UAV-centric communications. Reconfigurable holographic surface (RHS)-enabled APs provide programmable analog receive beams and generate scalar post-RHS observations, which are jointly processed at the CPU for distributed uplink MIMO detection at relatively low hardware cost. {Conventional AP association relies on distance, large-scale fading, or post-combining SINR obtained with user-specific digital combiners. Here, however, all UAVs served by a single-feed RHS AP share one amplitude-constrained receive pattern and one scalar AP output. We therefore derive an SINR-like score from this physical output and, under weak inter-AP disturbance correlation, approximate the CPU-side log-det objective by an additive per-AP surrogate, yielding a low-complexity ranking rule.} The results show that the nearest AP is not always the best choice, the AP-UAV height difference may have an optimal value, and larger serving clusters bring diminishing returns. Simulations show that the proposed method improves the minimum UAV data rate, average spectral efficiency, fairness, and energy efficiency compared with benchmark schemes.

\end{abstract}


\section{Introduction}
\IEEEPARstart{I}{ndoor} industrial UAV communication networks are expected to support emerging applications such as warehouse inspection, facility monitoring, inventory sensing, and production-line inspection. In such scenarios, UAVs usually operate inside workshops or warehouses filled with shelves, metal equipment, rails, robots, and conveyor systems, which lead to link-dependent shadowing, occasional blockage, and highly non-uniform large-scale propagation conditions. Moreover, inspection UAVs are often deployed in localized task regions and may move along similar straight-line paths at approximately fixed altitudes, rather than being uniformly distributed over the entire hall. Meanwhile, industrial uplink services typically place stronger emphasis on reliability, coverage continuity, link stability, and weak-UAV performance than on hotspot peak throughput. These features make conventional small-cell architectures less suitable for indoor industrial UAV communications.

Cell-free architectures are promising for such scenarios because distributed access points (APs) connected to a central processing unit (CPU) can jointly receive uplink signals from UAVs and thereby improve service robustness \cite{1Ngo2017CellFree,2Bjornson2020Making,wang2023task,4Demir2021Foundations,5Chen2022Survey,6Mahmoud2024Indoor,7Alonzo2021Reliable,8Zhang2022Indoor,wang2025efficient}. Since the seminal work on cell-free massive MIMO \cite{1Ngo2017CellFree}, subsequent studies have shown the importance of centralized processing, scalable implementation, and user-centric reception for improving performance and complexity tradeoffs \cite{2Bjornson2020Making,4Demir2021Foundations,5Chen2022Survey}. For indoor industrial UAV uplink transmission, UAV-centric reception is particularly attractive because it can alleviate the cell-edge effect, improve macro-diversity, and provide stable connectivity under blockage and uneven UAV distributions \cite{6Mahmoud2024Indoor,7Alonzo2021Reliable,8Zhang2022Indoor,10Alonzo2020URLLC}. However, the performance of a cell-free system depends critically on AP association, since each UAV must determine an appropriate serving AP set before CPU-side combining is performed.

Existing AP association methods are still largely based on distance, large-scale fading, received power, threshold rules, heuristic matching, mixed-integer optimization, or learning-based policies \cite{6Mahmoud2024Indoor,7Alonzo2021Reliable,8Zhang2022Indoor,10Alonzo2020URLLC,11Bjornson2020Scalable,Xue2025Handover,Kandil2025Greedy,Shi2023UAVThreshold,AlAlwani2025CAPA,DAndrea2020Hungarian,Liao2026EdgeEnhanced,Liu2026EGAT,
Zhou2025TDDPG,Du2025DMA,Shi2025SIM,Wang2022mmWaveCF}. These methods are useful for conventional cell-free and distributed MIMO systems, and prior studies on factory automation and indoor industrial networks have demonstrated the importance of vertical geometry, centralized processing, and user-centric AP selection \cite{6Mahmoud2024Indoor,7Alonzo2021Reliable,8Zhang2022Indoor, 10Alonzo2020URLLC}. Nevertheless, {many low-complexity rules among them} assume conventional phased-array (PA)-based AP transceivers and implicitly regard average link strength as the dominant indicator of AP quality.
This assumption becomes less accurate in indoor industrial UAV scenarios. First, UAVs may appear in clustered regions, which increases the probability of strong spatial overlap and localized interference. Second, UAVs following nearby inspection paths at similar heights may be difficult to separate from the viewpoint of a given AP, even when their large-scale fading values are favorable \cite{12Bjornson2018Unlimited,13GPP36814}. Therefore, a nearer AP or an AP with stronger average received power is not necessarily the best serving AP.

Reconfigurable holographic surfaces (RHSs), which are closely related to large intelligent surfaces, holographic MIMO, and dynamic metasurface antennas, provide densely spaced programmable elements, amplitude-only control, and flexible directional responses with relatively low hardware complexity \cite{14Dardari2020LIS,15Huang2020Holographic,16Shlezinger2021DMA,17Deng2021Future,18Deng2021Beamforming,19Deng2022HDMA,20Deng2022Amplitude,21Deng2023Paradigm,22Fong2010Scalar}. Existing RHS-related studies have mainly focused on hardware architectures, holographic beamforming, multi-UAV access, and holographic radio \cite{17Deng2021Future,18Deng2021Beamforming,19Deng2022HDMA,20Deng2022Amplitude,21Deng2023Paradigm,22Fong2010Scalar}. These features make RHSs attractive as AP front ends for indoor industrial UAV cell-free networks, since they can offer both large effective apertures and fine-grained directional selectivity for uplink MIMO reception. { Both single- and multi-feed RHS architectures have been investigated, while the single-feed leaky-wave RHS represents a low-complexity implementation with experimentally demonstrated feasibility \cite{21Deng2023Paradigm,22Fong2010Scalar}. Accordingly, this work focuses on a single-feed, single-RF-chain RHS receiver, which produces one scalar output per AP.} At the same time, {the considered single-feed} RHS-enabled reception further exposes the limitations of conventional AP association rules. In particular, the service capability of a candidate AP depends not only on large-scale fading, but also on the effective channel gain produced by the RHS response, the overlap between different UAV channels, and the post-combining disturbance covariance{, because all UAVs associated with an AP share the same amplitude-constrained receive pattern}. Therefore, AP association for uplink RHS-enabled MIMO indoor industrial cell-free networks remains insufficiently explored.

Motivated by these observations, this paper studies AP association for uplink RHS-enabled MIMO transmission in indoor industrial UAV cell-free networks. An uplink signal model is developed for RHS-enabled MIMO indoor industrial UAV cell-free networks, where multi-antenna UAVs transmit to distributed RHS-enabled APs and the AP-side analog-combined scalar outputs are forwarded to the CPU for combining. The goal is to develop an association metric that captures the actual service capability of each AP for a given UAV while remaining interpretable and suitable for low-complexity ranking-based association. Building on the effective desired gain, inter-UAV overlap, and disturbance covariance after AP-side RHS combining, we formulate the AP association problem and derive a score-based rule. This relation between the proposed design and the indoor industrial UAV communication scenario deserves emphasis. Conventional distance- or fading-based association fails here for two reasons: clustered UAVs on similar paths present nearly identical LSF and ranges to a given AP, so scalar ranking cannot distinguish APs with severe inter-UAV overlap, a limitation that the proposed score overcomes by directly sensing this interference in its denominator; and dense blockages create strong but angularly mismatched links whose multipath energy the RHS receive pattern cannot effectively combine, so evaluating the post-RHS signal quality naturally deprioritizes such APs. Moreover, the system model embeds the scenario geometry and operation: the three-dimensional AP-UAV geometry underpins the height-angle tradeoff, which exists only with ceiling-mounted APs and constant-altitude UAVs; the deterministic straight-line motion model matches aisle inspection and enables trajectory-predictable association updates; and the large shadowing variance captures the severe indoor environment and produces the distance-order reversal probability.

{This paper provides two} main contributions:
\begin{enumerate}[leftmargin=1.2em]
\item a low-complexity AP association method is proposed for
RHS-enabled UAV-centric cell-free reception.

Unlike distance-, LSF-, and received-power-based methods
\cite{Xue2025Handover,Kandil2025Greedy,Shi2023UAVThreshold},
and coverage-, load-, or pilot- schemes whose AP metrics remain
largely LSF-based
\cite{AlAlwani2025CAPA,Liao2026EdgeEnhanced},
the proposed score captures the post-RHS desired gain, residual
inter-UAV leakage, and filtered noise. In contrast to conventional
SINR-based association with user-specific digital combiners, each
single-feed, single-RF-chain RHS AP produces one scalar output, so
its associated UAVs share one amplitude-constrained receive pattern.
Compared with joint MINLP-, learning-, SIM-, DMA-, and mmWave-based
designs
\cite{Liu2026EGAT,Zhou2025TDDPG,Du2025DMA,Shi2025SIM,
Wang2022mmWaveCF},
the proposed method retains a one-shot ranking implementation without
repeatedly solving the joint association--RHS optimization problem.
Under fixed probing weights and weak cross-AP disturbance correlation,
it provides a lower-bound surrogate of the probing-stage log-det rate,
for which the Top-$J_k$ rule is optimal.

\item Three interpretable conclusions are obtained. First, a nearer AP is not necessarily better than a farther AP. Second, the AP-UAV relative height difference introduces a height-angle tradeoff, which implies the existence of an interior optimum under suitable conditions. Third, enlarging the serving cluster leads to diminishing marginal gains. 
\end{enumerate}

The rest of this paper is organized as follows. Section II presents the system model. Section III develops the AP association method. Section IV {compares with existing methods, Section V} provides analytical findings and design guidelines. Section {VI} presents simulation results. Section {VII}  concludes the paper.

\section{System Model}
\subsection{Industrial indoor UAV communication scenario}
\begin{figure}[h]
\centering
\includegraphics[width=0.48\textwidth]{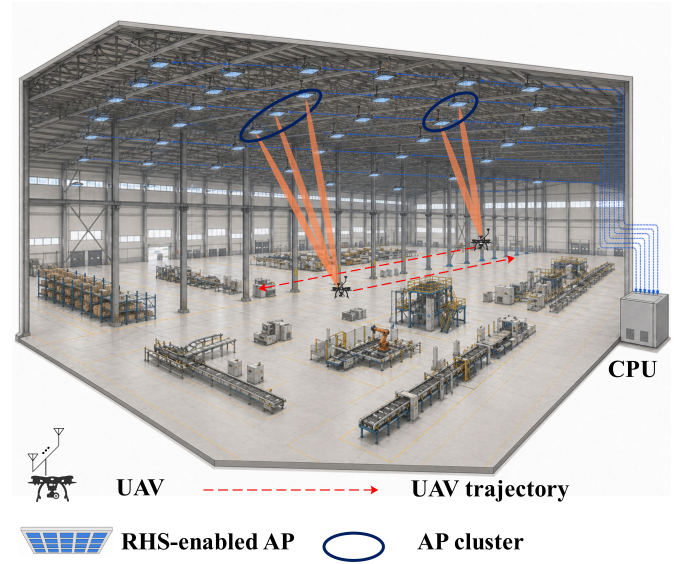}
\caption{Illustration of a cell-free industrial indoor UAV communication scenario.}
\label{fig:scenario}
\end{figure}

Consider an industrial indoor UAV communication scenario in a workshop or warehouse of size
$D_x^{\mathrm{hall}} \times D_y^{\mathrm{hall}} \times D_z^{\mathrm{hall}}$, as illustrated in Fig. \ref{fig:scenario},
where $L$ APs are connected to a CPU through wired links and jointly serve $K$ UAV users performing inspection and monitoring tasks.
Each RHS-enabled AP has a single RF chain and produces one analog-combined scalar observation per resource block. The MIMO dimension is formed by the multi-antenna UAV streams and the stacked scalar observations from multiple distributed APs at the CPU.

Each UAV follows a straight-line inspection trajectory inside the hall at an approximately constant altitude, which is consistent with aisle inspection in warehouses or facility inspection along shelves, beams, or production lines.
UAV user $k$ is equipped with $M_k$ antennas, while AP $l$ is equipped with an RHS consisting of $N_l$ controllable radiating elements.
Let $\mathbf{q}_l=[q_{l,x},q_{l,y},q_{l,z}]^{\top}\in\mathbb{R}^{3}$ denote the position of AP $l$, and let
$
\mathbf{Q}_k=
\left[
\mathbf{q}_{k,1},\ldots,\mathbf{q}_{k,M_k}
\right]
\in\mathbb{R}^{3\times M_k}
$
collect the antenna positions of UAV user $k$, where $\mathbf{q}_{k,m}\in\mathbb{R}^{3}$ is the position of the $m$-th UAV antenna.
For the RHS of AP $l$, let $\mathcal{N}_l=\{1,\ldots,N_l\}$ denote the element index set.
The $n$-th RHS element is characterized by its local surface coordinate
$\mathbf{u}_{l,n}\in\mathbb{R}^{2}$ and physical position
$\mathbf{r}_{l,n}\in\mathbb{R}^{3}$.
Define the element-position matrix
$
\mathbf{R}_l=
\left[
\mathbf{r}_{l,1},\ldots,\mathbf{r}_{l,N_l}
\right]
\in\mathbb{R}^{3\times N_l}.
$
The local coordinate $\mathbf u_{l,n}\in\mathbb R^2$ describes the position of the
$n$-th RHS element on the two-dimensional surface, whereas
$\mathbf r_{l,n}\in\mathbb R^3$ denotes its physical position in the global coordinate
system. For a planar RHS, 
$
\mathbf{u}_{l, n}=
\left[\begin{array}{c}
u_{l, n}^{(x)} \\
u_{l, n}^{(y)}
\end{array}\right]
$
is the two-dimensional local coordinate of the
$n$-th RHS element on the surface, if the RHS is horizontally mounted and aligned with the global $x$-$y$ plane, then
$
\mathbf{q}_l+\left[\begin{array}{c}
u_{l, n}^{(x)} \\
u_{l, n}^{(y)} \\
0
\end{array}\right]
$
Within one channel coherence block, each UAV is assumed to be quasi-static, so that its position and body orientation can be treated as approximately constant. The effects of residual UAV attitude variation and antenna misalignment are absorbed into the effective steering vectors and channel matrices defined below.

For UAV user $k$, let
$
\bar{\mathbf{q}}_k(t)
=
\frac{1}{M_k}\sum_{m=1}^{M_k}\mathbf{q}_{k,m}(t)
$
denote the UAV-center position at time $t$.
Since UAV $k$ performs inspection along a straight-line trajectory at a constant altitude, its center position is modeled as
\begin{equation}
\bar{\mathbf{q}}_k(t)
=
\mathbf{q}_{k,0}
+
v_k t\,\mathbf{d}_k,
\label{eq:uav_traj}
\end{equation}
where $\mathbf{q}_{k,0}\in\mathbb{R}^{3}$ is the initial UAV-center position, $v_k$ is the UAV speed, and 
$
\mathbf{d}_k=[d_{k,x},d_{k,y},0]^{\top}\in\mathbb{R}^{3}
$
is a constant unit-direction vector satisfying
$
\|\mathbf{d}_k\|_2=1.
$
Hence, the UAV flies at a fixed altitude along its inspection path.
For AP $l$ and UAV user $k$, define the relative position vector as
$
\boldsymbol{\delta}_{lk}(t)
=
\mathbf{q}_l-\bar{\mathbf{q}}_k(t)
=
\left[
\delta_{lk,x}(t),\delta_{lk,y}(t),\delta_{lk,z}(t)
\right]^{\top}.
$
Accordingly, the horizontal distance, the vertical height difference, and the three-dimensional distance between AP $l$ and UAV user $k$ are defined as
$
r_{lk}(t)
=
\sqrt{\delta_{lk,x}^{2}(t)+\delta_{lk,y}^{2}(t)},
$
$
z_{lk}
=
|\delta_{lk,z}(t)|,
$
$
d_{lk}(t)
=
\|\boldsymbol{\delta}_{lk}(t)\|_2
=
\sqrt{r_{lk}^{2}(t)+z_{lk}^{2}}.
$
Since the UAV altitude is fixed and AP $l$ is stationary, the vertical height difference $z_{lk}$ remains constant along the trajectory, whereas the horizontal distance $r_{lk}(t)$ and the three-dimensional distance $d_{lk}(t)$ vary with time.
If the RHS normal points vertically downward, the incidence angle $\theta_{lk}(t)$ satisfies
$
\cos\theta_{lk}(t)
=
\frac{z_{lk}}{d_{lk}(t)}.
$

\begin{figure}[h]
\centering
\includegraphics[width=0.46\textwidth]{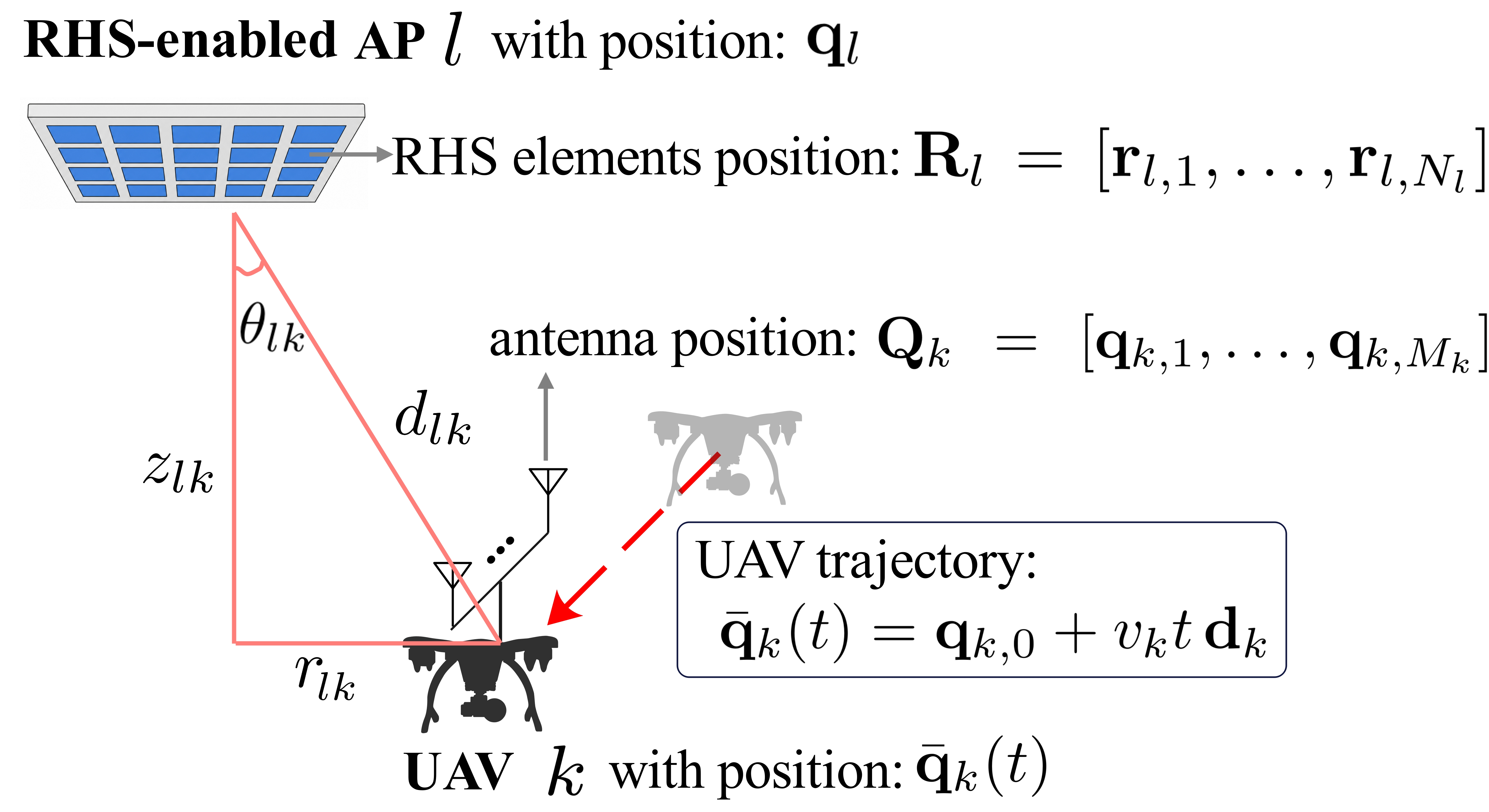}
\caption{Geometry of UAV-AP communications.}
\label{fig:geo}
\end{figure}

The geometry of the above UAV-AP pair is shown in Fig. \ref{fig:geo}. Following a three-dimensional indoor propagation model for aerial UAVs, the LSF is expressed as
\begin{equation}
\beta_{lk}
=
\bar{C}\,d_{lk}^{-\alpha}\chi_{lk},
\end{equation}
where $\alpha$ is the path-loss exponent, $\chi_{lk}$ is the shadow fading factor, and $\bar{C}$  denotes the channel power gain at a reference
distance of $1$ m, $\bar C=\left(\frac{\lambda}{4\pi d_0}\right)^2$, $d_0=1~\mathrm{m}.$
This model is used to capture the dominant distance-dependent attenuation for indoor UAV links, while the effects of UAV height, local blockage, and propagation heterogeneity are reflected through the three-dimensional geometry and the shadowing term.
For notational simplicity, the explicit time index $t$ is omitted, and all quantities are understood as being evaluated at an instantaneous point on the UAV inspection trajectory. Then, the above notations will be $r_{l,k}$, $\theta_{l,k}$, $d_{l,k}$.

\subsection{RHS receive model}
\begin{figure}[htbp]
\centering
\includegraphics[width=0.48\textwidth]{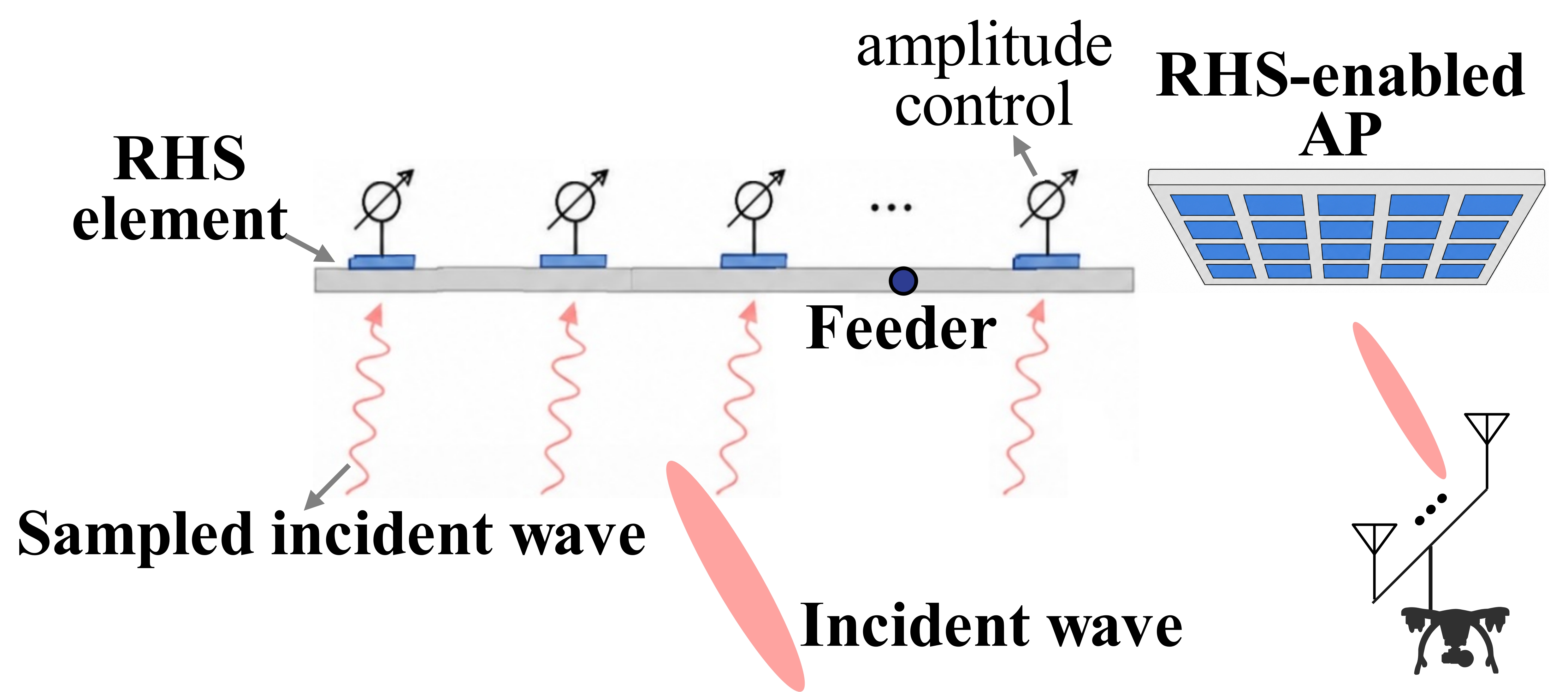}
\caption{
RHS-based uplink receiving. The incident wave is spatially sampled by the RHS elements, each of which applies an amplitude weight before coupling to the common feeder.
}
\label{fig:RHS}
\end{figure}
In transmit mode, the RHS functions as a holographic beamformer where a reference wave on the surface launched from the feed is shaped by pre‑designed real‑valued amplitude weights that encode the interference between the reference wave and the desired spatial object wave, and the resulting diffraction reconstructs the target beam in space. In receive mode, the same RHS operates as a passive amplitude‑only weighted array with no reference wave present; incident waves are sampled at each element, multiplied by the same amplitude coefficients, and then passively combined through a fixed feeder network into a single RF chain. The receive process merely reuses the pre‑computed transmit holographic pattern, exploiting channel reciprocity because the phase compensation that produces a directional beam in transmission also provides maximum‑ratio combining for signals arriving from that same direction.
Consider the uplink transmission from UAV $k$ to AP $l$, where the RHS is used as a receive antenna.
UAV $k$ transmits
$
\mathbf x_k
=
[x_{k,1},\ldots,x_{k,M_k}]^\top
\in
\mathbb C^{M_k\times 1}.
$
The covariance of $\mathbf x_k$ is
\begin{equation}
\mathbb E\{\mathbf x_k\mathbf x_k^H\}
=
\mathbf P_k,
\end{equation}
and for equal-power transmission,
$
\mathbf P_k
=
p_k\mathbf I_{M_k}.
$
The RHS surface consists of \(N_l\) elements. Each element independently receives the incident signals from all UAVs in space.
Each element applies a real-valued amplitude control. All elements are combined through a fixed feeder network to a single RF chain, producing a scalar output $z_l$.

\subsubsection{Step 1: Element-level received signal}
The signal $[\mathbf{y}_{l,k}]_n$ received by the \(n\)-th element of the RHS mounted on AP $l$  from UAV $k$ is the spatial sample of the incident waves:
\begin{equation}
\mathbf y_{l,k}
=
\mathbf H_{l,k}\mathbf x_k
+
\mathbf n_l
\in
\mathbb C^{N_l\times 1},
\label{eq:y_lk_rhs}
\end{equation}
with
\begin{equation}
\begin{aligned}
&\mathbf H_{l,k}
=
\sqrt{\beta_{lk}\cos \theta_{l k}}\\
&\begin{bmatrix}
e^{j\boldsymbol{\kappa}_{l,k,1}^{\top}\mathbf u_{l,1}} e^{-j \frac{2\pi}{\lambda} d_{l, k, 1}^{(0)}}
&
\cdots
&
e^{j\boldsymbol{\kappa}_{l,k,M_k}^{\top}\mathbf u_{l,1}} e^{-j \frac{2\pi}{\lambda} d_{l, k, M_k}^{(0)}}
\\
e^{j\boldsymbol{\kappa}_{l,k,1}^{\top}\mathbf u_{l,2}}  e^{-j \frac{2\pi}{\lambda} d_{l, k, 1}^{(0)}}
&
\cdots
&
e^{j\boldsymbol{\kappa}_{l,k,M_k}^{\top}\mathbf u_{l,2}} e^{-j \frac{2\pi}{\lambda} d_{l, k, M_k}^{(0)}}
\\
\vdots
&
\ddots
&
\vdots
\\
e^{j\boldsymbol{\kappa}_{l,k,1}^{\top}\mathbf u_{l,N_l}} e^{-j \frac{2\pi}{\lambda} d_{l, k, 1}^{(0)}}
&
\cdots
&
e^{j\boldsymbol{\kappa}_{l,k,M_k}^{\top}\mathbf u_{l,N_l}} e^{-j \frac{2\pi}{\lambda} d_{l, k, M_k}^{(0)}}
\end{bmatrix}
\end{aligned}
\label{eq:H_lk_target_uav}
\end{equation}
$\in
\mathbb C^{N_l\times M_k},$
where
\begin{itemize}
    \item \(e^{j\boldsymbol{\kappa}_{l,k,m}^{\top}\mathbf{u}_{l,n}}\) is the spatial phase sample of the incident wave at this element;
    \item \(e^{-j\frac{2\pi}{\lambda}d_{l,k,m}^{(0)}}\) is the absolute propagation phase correction;
\end{itemize}
and $
\mathbf n_l
\sim
\mathcal{CN}(\mathbf 0,\sigma_l^2\mathbf I_{N_l}).
$ 

Let $\boldsymbol{\kappa}_{l,k,m}\in\mathbb R^2$ denote the equivalent transverse wavevector associated with the uplink wave emitted by the $m$-th antenna of UAV $k$ and arriving at AP $l$. 
Since the RHS is horizontally mounted and aligned with the global $x$-$y$ plane, $\boldsymbol{\kappa}_{l,k,m}$ is obtained by projecting the propagation direction from $\mathbf q_{k,m}$ to $\mathbf q_l$ onto the local RHS plane. 
Under the plane-wave approximation across the RHS surface, we have
\begin{equation}
\boldsymbol{\kappa}_{l,k,m}
=-\frac{2\pi}{\lambda}
\left[
\begin{array}{c}
\dfrac{q_{l,x}-q_{k,m,x}}
{\|\mathbf q_l-\mathbf q_{k,m}\|_2}
\\[2mm]
\dfrac{q_{l,y}-q_{k,m,y}}
{\|\mathbf q_l-\mathbf q_{k,m}\|_2}
\end{array}
\right],
\label{eq:kappa_geometry}
\end{equation}
where $\lambda$ is the carrier wavelength. 
The negative sign follows the receive-phase convention in which the baseband phase of the received wave is proportional to the negative propagation distance. 
If the opposite phase convention is used, the sign of $\boldsymbol{\kappa}_{l,k,m}$ should be changed consistently in both the channel model and the holographic pattern construction. 
\begin{equation}
d_{l, k, m}^{(0)}=\left\|\mathbf{q}_l-\mathbf{q}_{k, m}\right\|_2
\end{equation} 
denotes the distance of the $m$-the antenna of UAV $k$ to AP $l$.

Note that before data transmission, the network operates in a pilot association phase.
During this phase, all APs estimate the spatial directions or effective
channels of candidate UAVs from uplink pilots and report the required
large-scale or post-RHS quality information to the CPU. Based on these
quantities, the CPU determines the UAV-centric serving AP sets
$\{\mathcal S_k\}_{k=1}^{K}$. Then, for each AP $l$, the set of UAVs
served by this AP is defined as
$
\mathcal U_l = \{ k \mid l \in \mathcal S_k \}.
$
The signal sampled across the RHS elements of AP $l$ becomes
$
\mathbf y_l
=
\sum_{k\in\mathcal U_l}
\mathbf H_{l,k}\mathbf x_k
+
\mathbf n_l
$, i.e.,
$
\left[\mathbf{y}_l\right]_n=\sum_{k\in\mathcal U_l} \sum_{m=1}^{M_k} \sqrt{\beta_{l k} \cos \theta_{l k}} \cdot e^{j \boldsymbol{\kappa}_{l, k, m}^{\top} \mathbf{u}_{l, n}}  e^{-j \frac{2 \pi}{\lambda} d_{l, k, m}^{(0)}}  x_{k, m}+n_{l, n},
$
with
\begin{equation}
\left[\mathbf H_{l,k}\right]_n= \sum_{m=1}^{M_k} \sqrt{\beta_{l k} \cos \theta_{l k}} \cdot e^{j \boldsymbol{\kappa}_{l, k, m}^{\top} \mathbf{u}_{l, n}}  e^{-j \frac{2 \pi}{\lambda} d_{l, k, m}^{(0)}} ,
\end{equation}
where $x_{k,m}$ is the symbol transmitted from the \(m\)-th antenna of UAV \(k\); $n_{l, n}$ is the thermal noise at this element.

\subsubsection{Step 2: Amplitude weighting}
Each element applies its real-valued amplitude gain \(b_{l,n}\) to the received signal before coupling to the feeder:
$
[\tilde{\mathbf{y}}_l]_n = b_{l,n} \cdot [\mathbf{y}_l]_n .
$
by tuning a local impedance (e.g., via a varactor). Fig. \ref{fig:RHS} illustrates the control process.
The  composite amplitude pattern is then
$
b_{l,n}
=
\sum_{k\in\mathcal U_l}
\sum_{m=1}^{M_k}
a_{l,k,m}\bar b_{l,k,m,n},
$
denote
$
\{a_{l,k,m}:k\in\mathcal U_l,m=1,\ldots,M_k\}
$  as $\mathbf A_l$,
$
a_{l,k,m}\ge 0,
$
and
$
\sum_{k\in\mathcal U_l}
\sum_{m=1}^{M_k}
a_{l,k,m}
=
1.
$
Since each basis pattern satisfies
$
0\le \bar b_{l,k,m,n}\le 1,
$
the convex combination also satisfies
$
0\le b_{l,n}\le 1
$; and
$
\bar{b}_{l, k, m, n}
=\frac{1+\cos \left(\boldsymbol{\kappa}_{l, k, m}^{\top} \mathbf{u}_{l, n}-\beta_s \rho_{l, n}\right)}{2}
=\frac{1+\Re\left(e^{j \boldsymbol{\kappa}_{l, k, m}^{\top} \mathbf{u}_{l, n}} e^{-j \beta_s \rho_{l, n}}\right)}{2} .
$
which maximizes the received power from the direction of $\boldsymbol{\kappa}_{l, k, m}$ when applied alone. $\rho_{l, n}$ is the propagation distance along the feeding structure, $\rho_{l, n}=\left\|\mathbf{u}_{l, n}-\mathbf{u}_{l, 0}\right\|_2=\sqrt{\left(u_{l, n}^{(x)}-u_{l, 0}^{(x)}\right)^2+\left(u_{l, n}^{(y)}-u_{l, 0}^{(y)}\right)^2}$, and $\beta_s$ is the guided-mode propagation constant. Any other pattern that concentrates gain toward the desired angular direction can be used without affecting the subsequent association algorithm.

\subsubsection{Step 3: Feeder combining}
The weighted signals are combined through a fixed feeder network toward the RF chain.
The feeder introduces a fixed transmission coefficient from the \(n\)-th element to the RF chain:
$
g_{l,n} = \frac{1}{\sqrt{N_l}}\sqrt{\eta_{l,n}} \cdot e^{-\alpha_s \rho_{l, n}} \cdot e^{-j \beta_s \rho_{l, n}}
$
which is the internal feed response from the $n$-th RHS element to the single RF port, where $\eta_{l, n}$ is the element-to-feed coupling efficiency, $\alpha_s$ is the guided-mode attenuation coefficient.

\subsubsection{Step 4: Scalar RF-chain output}
\begin{figure}[h]
\centering
\includegraphics[width=0.44\textwidth]{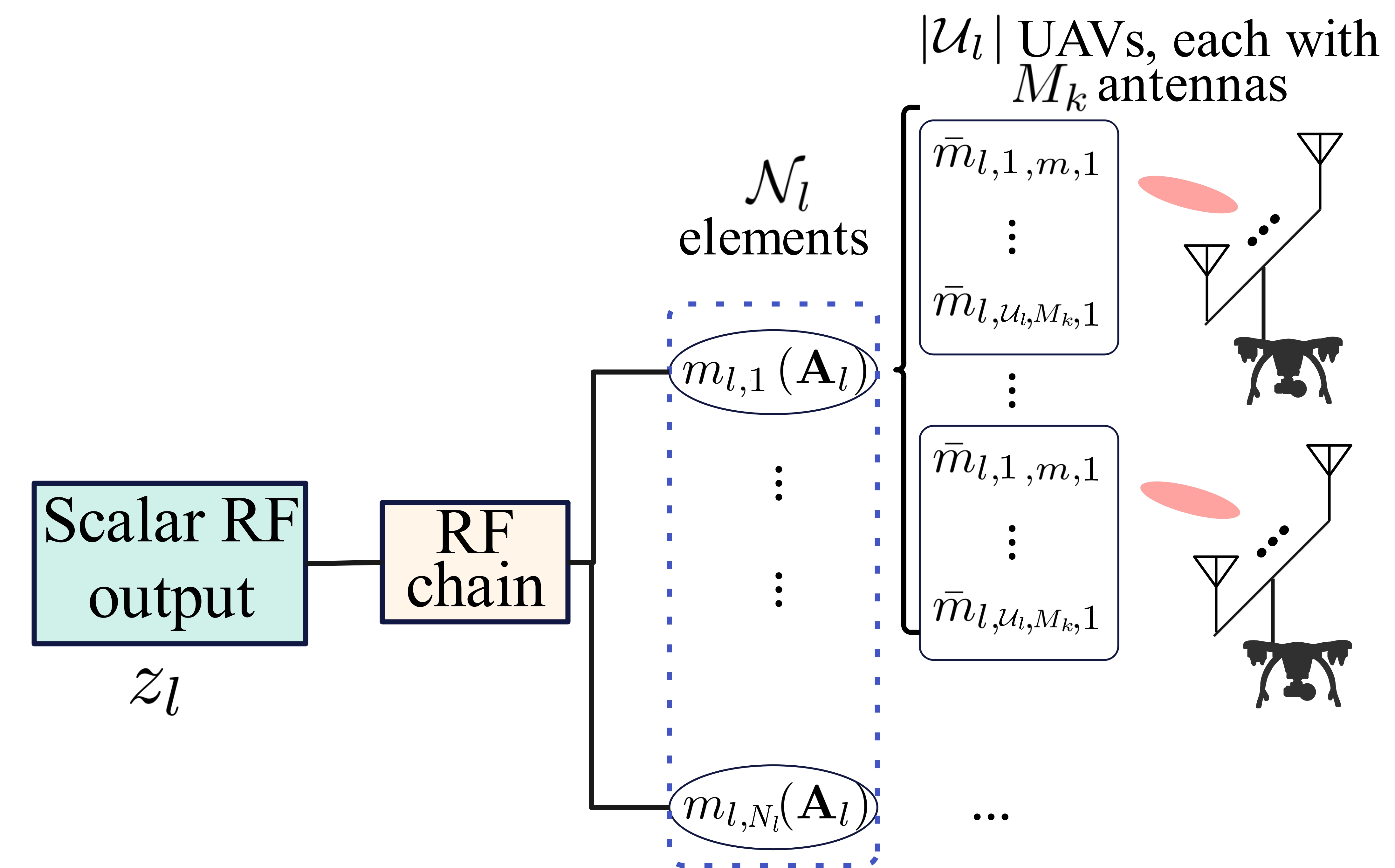}
\caption{
Composite receive amplitude pattern construction under a single-feed RHS architecture. 
}
\label{fig:receive_pattern_superposition}
\end{figure}

As shown in Fig. \ref{fig:receive_pattern_superposition}, 
define the physical feed-combining vector
$\mathbf v_l\in\mathbb C^{N_l\times 1}$ with entries
\begin{equation}
[\mathbf v_l]_n
=
g_{l,n} b_{l,n}
=
\frac{1}{\sqrt{N_l}}
\sqrt{\eta_{l,n}}\,
b_{l,n}
e^{-\alpha_s\rho_{l,n}}
e^{-j\beta_s\rho_{l,n}} .
\end{equation}
The scalar RF-chain output is therefore
\begin{equation}
z_l
=
\mathbf v_l^{\top}\mathbf y_l
=
\sum_{n=1}^{N_l}g_{l,n}b_{l,n}[\mathbf y_l]_n .
\end{equation}
For consistency with the standard complex-baseband inner-product notation,
we define the receive combining vector as
$
\mathbf w_l \triangleq \mathbf v_l^{*}.
$
Then the same scalar output can be written as
$
z_l=\mathbf w_l^H\mathbf y_l .
$
Accordingly,
$
[\mathbf w_l]_n^*
=[\mathbf v_l]_n
$
Then the scalar output can be expanded as \eqref{eq:z_l_}.
\begin{figure*}
\begin{equation}
\begin{aligned}
z_l
&=
\underbrace{\mathbf w_l^H\mathbf y_l}_{\text{Scalar RF-chain output for AP $l$}}
=
\underbrace{
\sum_{k\in\mathcal U_l}
\mathbf w_l^H\mathbf H_{l,k}\mathbf x_k
}_{\text{AP $l$ received from served UAVs}}
+
\underbrace{
\mathbf w_l^H\mathbf n_l
}_{\text{Filtered noise}}
\\
&=
\sum_{k\in\mathcal U_l}
\sum_{m=1}^{M_k}
\underbrace{
\left(
\sqrt{\beta_{lk}\cos \theta_{l k}}
\sum_{n=1}^{N_l}
[\mathbf w_l]_n^*
e^{j\boldsymbol{\kappa}_{l,k,m}^{\top}\mathbf u_{l,n}} 
e^{-j \frac{2 \pi}{\lambda} d_{l, k, m}^{(0)}}
\right)
}_{\text{Effective post-RHS channel from stream }m\text{ of UAV }k\text{ to AP }l}
x_{k,m}
+
\underbrace{
\sum_{n=1}^{N_l}
[\mathbf w_l]_n^*
n_{l,n}
}_{\text{Sum of thermal noises from all RHS elements in AP $l$ }}
\\
&=
\frac{1}{\sqrt{N_l}}
\sum_{n=1}^{N_l}
\underbrace{
\sqrt{\eta_{l,n}}
}_{\text{Element radiation efficiency}}
\underbrace{
\sum_{k\in\mathcal U_l}
\sum_{m=1}^{M_k}
a_{l,k,m}
\frac{
1+\Re\left(
e^{j \boldsymbol{\kappa}_{l, k, m}^{\top} \mathbf{u}_{l, n}}
e^{-j \beta_s \rho_{l, n}}
\right)
}{2}
}_{\text{Amplitude-controlled receive holographic pattern}}
\underbrace{
e^{-\alpha_s\rho_{l,n}}
}_{\text{Feed attenuation}}
\underbrace{
e^{-j \beta_s \rho_{l, n}}
}_{\text{Feed phase shift}}
\\
&\quad \times
\Bigg[
\underbrace{
\sum_{i\in\mathcal U_l}
\sum_{m=1}^{M_i}
\underbrace{
\sqrt{\beta_{li}\cos\theta_{li}}
}_{\text{Projected LSF}}
\underbrace{
e^{j\boldsymbol{\kappa}_{l,i,m}^{\top}\mathbf{u}_{l,n}}
}_{\text{Element sampling phase}}
\underbrace{
e^{-j\frac{2\pi}{\lambda}d_{l,i,m}^{(0)}}
}_{\text{Propagation phase}}
x_{i,m}
}_{\text{Signal received at element } n \text{ from all served UAVs}}
+
\underbrace{
n_{l,n}
}_{\text{Element thermal noise}}
\Bigg].
\end{aligned}
\label{eq:z_l_}
\end{equation}
\hrulefill
\end{figure*}
Denote the noise as
$
\tilde n_l
=
\mathbf w_l^H\mathbf n_l
=
\sum_{n=1}^{N_l}
[\mathbf w_l]_n^* n_{l,n},
$
and its variance $
\mathbb E\{|\tilde n_l|^2\}
=
\sigma_l^2
\mathbf w_l^H\mathbf w_l
=
\sigma_l^2
\sum_{n=1}^{N_l}
|[\mathbf w_l]_n|^2.
$
For UAV $k$, let
$
\mathcal S_k
=
\{l_1,\ldots,l_{|\mathcal S_k|}\}
$
denote the serving AP set whose scalar outputs are used by the CPU to
detect
$
\mathbf x_k\in\mathbb C^{M_k\times 1}.
$
The CPU stacks the scalar outputs of the selected APs as
$
\mathbf z_{\mathcal S_k}
=
[
z_{l_1},\ldots,z_{l_{|\mathcal S_k|}}
]^\top
\in
\mathbb C^{|\mathcal S_k|\times 1}.
$

For any UAV $i$, define the stacked post-RHS effective channel over the
serving AP set $\mathcal S_k$ as
\begin{equation}
\mathbf G_{k,i}(\mathcal S_k)
=
\begin{bmatrix}
\mathbf w_{l_1}^H\mathbf H_{l_1,i}\\
\mathbf w_{l_2}^H\mathbf H_{l_2,i}\\
\vdots\\
\mathbf w_{l_{|\mathcal S_k|}}^H\mathbf H_{l_{|\mathcal S_k|},i}
\end{bmatrix}
\in
\mathbb C^{|\mathcal S_k|\times M_i},
\end{equation}
where
$
\mathbf w_l\in\mathbb C^{N_l\times 1}
$
and
$
\mathbf H_{l,i}\in\mathbb C^{N_l\times M_i}
$.
Hence,
$
\mathbf w_l^H\mathbf H_{l,i}\in\mathbb C^{1\times M_i}
$.
The $m$-th entry of $\mathbf w_l^H\mathbf H_{l,i}$ is
$
[
\mathbf w_l^H\mathbf H_{l,i}
]_m
=
\sqrt{\beta_{li}\cos \theta_{li}}
\sum_{n=1}^{N_l}
[\mathbf w_l]_n^*
e^{j\boldsymbol{\kappa}_{l,i,m}^{\top}\mathbf u_{l,n}}
e^{-j \frac{2 \pi}{\lambda} d_{l,i,m}^{(0)}},
$
$m=1,\ldots,M_i$.
The stacked post-RHS noise vector is
$
\tilde{\mathbf n}(\mathcal S_k)
=
[
\tilde n_{l_1},\tilde n_{l_2},\ldots,
\tilde n_{l_{|\mathcal S_k|}}
]^\top
\in
\mathbb C^{|\mathcal S_k|\times 1}.
$
Using \eqref{eq:z_l_}, the stacked CPU observation for detecting UAV $k$
is
\begin{align}
\mathbf z_{\mathcal S_k}
=
\mathbf G_{k,k}(\mathcal S_k)\mathbf x_k
+
\sum_{i\neq k}
\mathbf G_{k,i}(\mathcal S_k)\mathbf x_i
+
\tilde{\mathbf n}(\mathcal S_k),
\end{align}
where
$
\mathbf G_{k,k}(\mathcal S_k)\mathbf x_k
\in\mathbb C^{|\mathcal S_k|\times 1}
$
and
$
\mathbf G_{k,i}(\mathcal S_k)\mathbf x_i
\in\mathbb C^{|\mathcal S_k|\times 1}.
$
The disturbance covariance for detecting UAV $k$ is
\begin{align}
\mathbf \Xi_k(\mathcal S_k)
=
&
\mathbb E
\Bigg[
\left(
\sum_{i\neq k}
\mathbf G_{k,i}(\mathcal S_k)\mathbf x_i
+
\tilde{\mathbf n}(\mathcal S_k)
\right)
\\
&
\left(
\sum_{i\neq k}
\mathbf G_{k,i}(\mathcal S_k)\mathbf x_i
+
\tilde{\mathbf n}(\mathcal S_k)
\right)^H
\Bigg]
\end{align}
$
\in
\mathbb C^{|\mathcal S_k|\times |\mathcal S_k|},
$
which is decomposed as 
\begin{align}
\sum_{i\neq k}
\mathbf G_{k,i}(\mathcal S_k)
\mathbf P_i
\mathbf G_{k,i}^{H}(\mathcal S_k)
+
\operatorname{diag}
\left(
\ldots,
\sigma_{l_{|\mathcal S_k|}}^{2}
\|\mathbf w_{l_{|\mathcal S_k|}}\|_2^2
\right)
\end{align}
where
$
\mathbf G_{k,i}(\mathcal S_k)\in\mathbb C^{|\mathcal S_k|\times M_i}
$,
$
\mathbf P_i\in\mathbb C^{M_i\times M_i}
$,
and
$
\mathbf G_{k,i}^{H}(\mathcal S_k)\in\mathbb C^{M_i\times |\mathcal S_k|}
$,
so that
$
\mathbf G_{k,i}(\mathcal S_k)
\mathbf P_i
\mathbf G_{k,i}^{H}(\mathcal S_k)
\in
\mathbb C^{|\mathcal S_k|\times |\mathcal S_k|}.
$ 
Since $\tilde{n}_l=\mathbf{w}_l^H \mathbf{n}_l$, if $\mathbf{n}_l \sim \mathcal{C} \mathcal{N}\left(\mathbf{0}, \sigma_l^2 \mathbf{I}_{N_l}\right)$, then 
\begin{align}
\mathbb{E}\left\{\left|\tilde{n}_l\right|^2\right\}=\mathbb{E}\left\{\mathbf{w}_l^H \mathbf{n}_l \mathbf{n}_l^H \mathbf{w}_l\right\}=\mathbf{w}_l^H \sigma_l^2 \mathbf{I}_{N_l} \mathbf{w}_l=\sigma_l^2\left\|\mathbf{w}_l\right\|_2^2 .
\end{align}
And, 
$\left\|\mathbf{w}_{l_j}\right\|_2^2=\sum_{n=1}^{N_{l_j}}\left|\left[\mathbf{w}_{l_j}\right]_n\right|^2$, therefore expand form is $\sigma_{l_j}^2 \sum_{n=1}^{N_{l_j}}\left|\left[\mathbf{w}_{l_j}\right]_n\right|^2=\sigma_{l_j}^2\left\|\mathbf{w}_{l_j}\right\|_2^2$ .
Assuming Gaussian signaling and treating inter-UAV overlap as colored
noise, the achievable information rate of UAV $k$ is given by
\eqref{eq:rate_logdet_expanded}.

\subsubsection{Step 5: CPU linear combining}

\begin{figure*}[t]
\begin{align}
R_k(\mathcal S_k)
&=
\log_2\det
\left(
\mathbf I_{M_k}
+
\mathbf P_k^{1/2}
\mathbf G_{k,k}^{H}(\mathcal S_k)
\mathbf \Xi_k^{-1}(\mathcal S_k)
\mathbf G_{k,k}(\mathcal S_k)
\mathbf P_k^{1/2}
\right)
\\
&=
\log_2\det
\left(
\mathbf I_{M_k}
+
\mathbf P_k^{1/2}
\begin{bmatrix}
\mathbf w_{l_1}^{H}\mathbf H_{l_1,k}\\
\mathbf w_{l_2}^{H}\mathbf H_{l_2,k}\\
\vdots\\
\mathbf w_{l_{|\mathcal S_k|}}^{H}\mathbf H_{l_{|\mathcal S_k|},k}
\end{bmatrix}^{H}
\mathbf \Xi_k^{-1}(\mathcal S_k)
\begin{bmatrix}
\mathbf w_{l_1}^{H}\mathbf H_{l_1,k}\\
\mathbf w_{l_2}^{H}\mathbf H_{l_2,k}\\
\vdots\\
\mathbf w_{l_{|\mathcal S_k|}}^{H}\mathbf H_{l_{|\mathcal S_k|},k}
\end{bmatrix}
\mathbf P_k^{1/2}
\right)
\\
&=
\log_2\det
\left(
\mathbf I_{M_k}
+
\mathbf P_k^{1/2}
\mathbf G_{k,k}^{H}(\mathcal S_k)
\left[
\sum_{i\ne k}
\mathbf G_{k,i}(\mathcal S_k)
\mathbf P_i
\mathbf G_{k,i}^{H}(\mathcal S_k)
+
\operatorname{diag}
\left(
\sigma_{l_1}^{2}\|\mathbf w_{l_1}\|_2^2,
\ldots,
\sigma_{l_{|\mathcal S_k|}}^{2}
\|\mathbf w_{l_{|\mathcal S_k|}}\|_2^2
\right)
\right]^{-1}
\mathbf G_{k,k}(\mathcal S_k)
\mathbf P_k^{1/2}
\right)
\\
&=
\log_2\det
\left(
\mathbf I_{M_k}
+
\mathbf P_k^{1/2}
\mathbf G_{k,k}^{H}(\mathcal S_k)
\left[
\sum_{i\ne k}
\begin{bmatrix}
\mathbf w_{l_1}^{H}\mathbf H_{l_1,i}\\
\mathbf w_{l_2}^{H}\mathbf H_{l_2,i}\\
\vdots\\
\mathbf w_{l_{|\mathcal S_k|}}^{H}\mathbf H_{l_{|\mathcal S_k|},i}
\end{bmatrix}
\mathbf P_i
\begin{bmatrix}
\mathbf w_{l_1}^{H}\mathbf H_{l_1,i}\\
\mathbf w_{l_2}^{H}\mathbf H_{l_2,i}\\
\vdots\\
\mathbf w_{l_{|\mathcal S_k|}}^{H}\mathbf H_{l_{|\mathcal S_k|},i}
\end{bmatrix}^{H}
+
\operatorname{diag}
\left(
\sigma_{l_1}^{2}\|\mathbf w_{l_1}\|_2^2,
\ldots,
\sigma_{l_{|\mathcal S_k|}}^{2}
\|\mathbf w_{l_{|\mathcal S_k|}}\|_2^2
\right)
\right]^{-1}
\right.
\\
&\qquad
\left.
\times
\begin{bmatrix}
\mathbf w_{l_1}^{H}\mathbf H_{l_1,k}\\
\mathbf w_{l_2}^{H}\mathbf H_{l_2,k}\\
\vdots\\
\mathbf w_{l_{|\mathcal S_k|}}^{H}\mathbf H_{l_{|\mathcal S_k|},k}
\end{bmatrix}
\mathbf P_k^{1/2}
\right),
\label{eq:rate_logdet_expanded}
\end{align}
\hrulefill
\end{figure*}
The condition
$
|\mathcal S_k|\ge M_k
$
is necessary but not sufficient for separating the $M_k$ streams.
A sufficient rank condition is
$
\operatorname{rank}
\left(
\mathbf G_{k,k}(\mathcal S_k)
\right)
=
M_k .
$
If this condition is not satisfied, the effective number of reliably
separable spatial streams is limited by the rank of the stacked
post-RHS channel matrix.
The CPU performs linear minimum mean square error (LMMSE) estimation as
$
\hat{\mathbf x}_k
=
\mathbf C_k^H
\mathbf z_{\mathcal S_k},
$
where
$
\hat{\mathbf x}_k\in\mathbb C^{M_k\times 1}
$,
$
\mathbf z_{\mathcal S_k}\in\mathbb C^{|\mathcal S_k|\times 1}
$,
and hence
$
\mathbf C_k^H\in\mathbb C^{M_k\times |\mathcal S_k|}
$.
The LMMSE combining matrix is
$
\mathbf C_k
=
\Big(
\mathbf G_{k,k}(\mathcal S_k)
\mathbf P_k
\mathbf G_{k,k}^{H}(\mathcal S_k)
+
\mathbf \Xi_k(\mathcal S_k)
\Big)^{-1}
\mathbf G_{k,k}(\mathcal S_k)
\mathbf P_k
$
$\in
\mathbb C^{|\mathcal S_k|\times M_k}$ gives the practical LMMSE estimator. 
Equivalently,
$
\mathbf C_k^H
=
\mathbf P_k
\mathbf G_{k,k}^{H}(\mathcal S_k)
\Big(
\mathbf G_{k,k}(\mathcal S_k)
\mathbf P_k
\mathbf G_{k,k}^{H}(\mathcal S_k)
+
\mathbf \Xi_k(\mathcal S_k)
\Big)^{-1}
$
$\in \mathbb C^{M_k\times |\mathcal S_k|}.$
Here, $
\mathbf P_k^{1/2}\in\mathbb C^{M_k\times M_k},
$
$
\mathbf G_{k,k}(\mathcal S_k)\in\mathbb C^{|\mathcal S_k|\times M_k},
$
$
\mathbf G_{k,k}^{H}(\mathcal S_k)\in\mathbb C^{M_k\times |\mathcal S_k|},
$
$
\mathbf \Xi_k(\mathcal S_k)\in\mathbb C^{|\mathcal S_k|\times |\mathcal S_k|},
$
$
\mathbf P_k^{1/2}
\mathbf G_{k,k}^{H}(\mathcal S_k)
\mathbf \Xi_k^{-1}(\mathcal S_k)
\mathbf G_{k,k}(\mathcal S_k)
\mathbf P_k^{1/2}
\in
\mathbb C^{M_k\times M_k}.
$

The above model can be extended to account for practical hardware impairments and power consumption as follows:
\begin{itemize}
\item \emph{Feeder loss:} considered by $\eta_{l,n}$ and
$e^{-\alpha_s\rho_{l,n}}$; the default values
$\eta_{l,n}=1$ and $\alpha_s=0$ isolate the association effect.

\item \emph{Mutual coupling:}
replace the effective channel by
$\mathbf H_{l,k}$ to $\mathbf C_l\mathbf H_{l,k}$
\cite{Zeng2026MutualCouplingISAC,Wu2026CouplingAwareRHS,
Guo2026MutualCouplingHybrid,Guo2025FastHash}, with $\mathbf C_l$ denotes the coupling matrix.

\item \emph{Finite-resolution and pattern errors:}
use $b_{l,n}$ to $\mathcal Q(b_{l,n})$, where $\mathcal Q(\cdot)$ maps the designed amplitude to the nearest
hardware-supported level, e.g., the ON/OFF states of PIN-diode
elements\cite{Deng2022HolographicRadio}. Fabrication, temperature,
and calibration errors can also be included by replacing the ideal
feeder coefficient $g_{l,n}$ with its actual value
$\widehat g_{l,n}$.

\item \emph{RF-chain impairments:}
include signal attenuation and an additional distortion term
\cite{Li2024RHSImpairments}, e.g., $\widehat z_l=\sqrt{\varepsilon_l}\,z_l+d_l$, where $\varepsilon_l\in[0,1]$ is the RF-chain hardware-quality
factor and $d_l$ is the signal-dependent distortion noise.

\item \emph{Power consumption:}
a full metric should include AP, RF-chain,
RHS-control, fronthaul, and CPU power. In this paper, we only consider the power of UAV transmitters.
\end{itemize}
These modify the effective post-RHS gain and disturbance.

\section{AP Association Method}
\subsection{Problem formulation}
The original
AP association problem is  {to maximize the CPU-side achievable rate in
\eqref{eq:rate_logdet_expanded},}
\begin{align}
(\mathcal P{0}):
\max_{\{\mathcal S_k\}_{k=1}^{K}}
\quad
&
\sum_{k=1}^{K} R_k(\mathcal S_k)
\\
\mathrm{s.t.}\quad
&
\mathcal S_k\subseteq\{1,\ldots,L\},\quad \forall k,
\\
&
M_k\leq |\mathcal S_k|\leq J_k,\quad \forall k .
\label{eq:P1_original}
\end{align}
Here $R_k(\mathcal S_k)$, defined in \eqref{eq:rate_logdet_expanded}, is the log-det achievable rate in
the system model per bandwidth, i.e., spectral efficiency (unit: bit/s/Hz). Although $(\mathcal P{0})$ is written in set
form, it is equivalent to a binary AP association problem.
$(\mathcal P{0})$ is hard to solve directly. For
$\mathcal S_k=\{l_1,\ldots,l_{|\mathcal S_k|}\}$, the
\((p,q)\)-th element of the disturbance covariance matrix is
$
[\mathbf \Xi_k(\mathcal S_k)]_{p,q}
=
\sum_{i\neq k}
\mathbf w_{l_p}^{H}
\mathbf H_{l_p,i}\mathbf P_i
\mathbf H_{l_q,i}^{H}
\mathbf w_{l_q}
+
\begin{cases}
\sigma_{l_p}^{2}\mathbf w_{l_p}^{H}\mathbf w_{l_p},
& p=q,\\
0, & p\neq q .
\end{cases}
$
Thus, different selected APs are coupled inside
$\mathbf \Xi_k(\mathcal S_k)$ and its inverse in the log-det
rate. {In addition, AP association and RHS configuration are mutually
dependent. Specifically, the serving AP sets determine
$
\{\mathcal S_k\}_{k=1}^{K}
$, then affects
$
\mathcal U_l=\{k:l\in\mathcal S_k\}
$, changes
$
\mathbf A_l
$
and
$
\mathbf w_l(\mathbf A_l),
$
whereas the receive pattern $\mathbf w_l(\mathbf A_l)$ is itself
required to evaluate the quality of AP $l$ before association.
To break this circular dependence while retaining a low-complexity
online implementation, we adopt a two-stage strategy, summarized as Algorithm \ref{alg:proposed}. }

\subsection{{Pilot association phase}}
Since the associated-user set $\mathcal U_l$ is unknown before AP
association, the data-phase RHS coefficient vector $\mathbf A_l^\star$
cannot be used to rank candidate APs.  Therefore, the association stage
uses a fixed $\mathbf w_l^{(0)}$, e.g., a wide-beam
or codebook-based probing pattern, which is independent of
$\mathcal U_l$. The preliminary association score is defined as
\begin{equation}
\psi_{l,k}^{(0)}
=
\frac{
(\mathbf w_l^{(0)})^{H}
\mathbf H_{l,k}\mathbf P_k\mathbf H_{l,k}^{H}
\mathbf w_l^{(0)}
}{
(\mathbf w_l^{(0)})^{H}
\left(
\sum_{i\neq k}
\mathbf H_{l,i}\mathbf P_i\mathbf H_{l,i}^{H}
+
\sigma_l^2\mathbf I_{N_l}
\right)
\mathbf w_l^{(0)}
}.
\label{eq:psi_asso}
\end{equation}

{
We consider
the weak cross-AP disturbance-correlation approximation, under which
the off-diagonal entries of $\mathbf\Xi_k(\mathcal S_k)$ are
neglected. The CPU-side rate becomes
$
R_k(\mathcal S_k)
\simeq
\log_2\det\Bigg(
\mathbf I_{M_k}
+
\sum_{l\in\mathcal S_k}
\frac{
\mathbf P_k^{1/2}
\mathbf H_{l,k}^{H}\mathbf w_l^{(0)}
(\mathbf w_l^{(0)})^{H}
\mathbf H_{l,k}\mathbf P_k^{1/2}
}{
(\mathbf w_l^{(0)})^{H}
\left(
\sum_{i\neq k}
\mathbf H_{l,i}\mathbf P_i\mathbf H_{l,i}^{H}
+
\sigma_l^2\mathbf I_{N_l}
\right)
\mathbf w_l^{(0)}
}
\Bigg).
$
}

{
\begin{proposition}[Trace-based bounds]
Under a fixed $\mathbf w_l^{(0)}$ and the weak cross-AP disturbance-correlation diagonal approximation, the CPU-side rate satisfies
\begin{align}
\log_2\left(
1+\sum_{l\in\mathcal S_k}\psi_{l,k}^{(0)}
\right)
\leq
R_k(\mathcal S_k)
\leq
M_k\log_2\left(
1+\frac{1}{M_k}
\sum_{l\in\mathcal S_k}\psi_{l,k}^{(0)}
\right).
\label{eq:trace_rate_bounds}
\end{align}
The two bounds coincide when $M_k=1$.
\end{proposition}}
{
\begin{proof}
Under the diagonal approximation, define
$
\mathbf A_k(\mathcal S_k)
=
\sum_{l\in\mathcal S_k}
\frac{
\mathbf P_k^{1/2}\mathbf H_{l,k}^{H}\mathbf w_l^{(0)}
(\mathbf w_l^{(0)})^{H}\mathbf H_{l,k}\mathbf P_k^{1/2}
}{
(\mathbf w_l^{(0)})^{H}
\left(
\sum_{i\neq k}\mathbf H_{l,i}\mathbf P_i\mathbf H_{l,i}^{H}
+\sigma_l^2\mathbf I_{N_l}
\right)
\mathbf w_l^{(0)}
}.
$
Each numerator has the form $\mathbf x\mathbf x^H$ and each
denominator is strictly positive due to thermal noise; hence,
$\mathbf A_k(\mathcal S_k)\succeq\mathbf 0$. By the linearity and
cyclic property of the trace,
\begin{align}
\operatorname{tr}\!\left(\mathbf A_k(\mathcal S_k)\right)
&=
\sum_{l\in\mathcal S_k}
\frac{
(\mathbf w_l^{(0)})^{H}
\mathbf H_{l,k}\mathbf P_k\mathbf H_{l,k}^{H}
\mathbf w_l^{(0)}
}{
(\mathbf w_l^{(0)})^{H}
\left(
\sum_{i\neq k}\mathbf H_{l,i}\mathbf P_i\mathbf H_{l,i}^{H}
+\sigma_l^2\mathbf I_{N_l}
\right)
\mathbf w_l^{(0)}
}
\nonumber\\
&=
\sum_{l\in\mathcal S_k}\psi_{l,k}^{(0)}.
\end{align}
Let $\lambda_1,\ldots,\lambda_{M_k}\geq0$ denote the eigenvalues of
$\mathbf A_k(\mathcal S_k)$. Then
\begin{align}
\det\!\left(\mathbf I_{M_k}+\mathbf A_k(\mathcal S_k)\right)
&=
\prod_{m=1}^{M_k}(1+\lambda_m)
\geq
1+\sum_{m=1}^{M_k}\lambda_m,
\end{align}
which gives the lower bound after taking $\log_2(\cdot)$. Moreover,
the arithmetic--geometric mean inequality yields
\begin{align}
\prod_{m=1}^{M_k}(1+\lambda_m)
\leq
\left(
1+\frac{1}{M_k}\sum_{m=1}^{M_k}\lambda_m
\right)^{M_k},
\end{align}
which gives the upper bound. Since
$\sum_m\lambda_m=\operatorname{tr}(\mathbf A_k(\mathcal S_k))
=\sum_{l\in\mathcal S_k}\psi_{l,k}^{(0)}$, the result follows.
When $M_k=1$, the matrix has only one eigenvalue, and both inequalities
hold with equality.
\end{proof}
}

{The lower bound in \eqref{eq:trace_rate_bounds} leads to the following separable surrogate of $(\mathcal P0)$:}
\begin{align}
(\mathcal P{0}'):
\max_{\mathcal S_k}
\quad
&
\Gamma_k(\mathcal S_k)
=
\sum_{l\in\mathcal S_k}\psi^{(0)}_{l,k}
\\
\mathrm{s.t.}\quad
&
M_k\le |\mathcal S_k|\le J_k .
\end{align}

\begin{proposition}[Optimal solution of \((\mathcal P{0}')\)]
Assume \(M_k\le |\mathcal S_k| \le J_k\leq L\). Since \(\psi^{(0)}_{l,k}\geq0\), an optimal solution is obtained by selecting the \(J_k\) APs with the largest scores:
$
\mathcal S_k^{\star}
=
\left\{
l:\psi_{l,k}\text{ is among the largest }J_k\text{ values for UAV }k
\right\}.
$
The optimal value is
$
\Gamma_k^{\star}
=
\sum_{l\in\mathcal S_k^{\star}}\psi^{(0)}_{l,k}.
$
If ties occur at the \(J_k\)-th largest score, any tied selection is optimal.
\end{proposition}

\begin{proof}
Because $\log_2(1+x)$ is monotonically increasing and
$\psi_{l,k}^{(0)}\geq0$, maximizing
$\Gamma_k(\mathcal S_k)$ is equivalent to selecting the $J_k$
largest scores.
\end{proof}

\subsection{{Data-phase RHS configuration}}
\label{sec:data_rhs_configuration}
{
The probing stage determines the serving AP sets
$\{\mathcal S_k^\star\}_{k=1}^{K}$. Consequently, the associated-UAV
set of AP $l$ becomes known as
\begin{equation}
\mathcal U_l=\{k:l\in\mathcal S_k^\star\}.
\end{equation}
Each AP can then configure its shared RHS receive pattern for its
selected UAVs. This second stage affects} the data-phase transmission
and final rate evaluation, but does not trigger AP re-ranking within
the same association interval.

For given $\mathbf A_l=\{a_{l,k,m}:k\in\mathcal U_l,m=1,\ldots,M_k\}$, define the post-RHS signal-to-disturbance
ratio of AP $l$ for UAV $k$:
\begin{align}
&\gamma_{l,k}(\mathbf A_l)\\
&=\frac{
\mathbf w_l^H(\mathbf A_l)
\mathbf H_{l,k}\mathbf P_k\mathbf H_{l,k}^{H}
\mathbf w_l(\mathbf A_l)
}{
\mathbf w_l^H(\mathbf A_l)
\left(
\sum_{i\neq k}
\mathbf H_{l,i}\mathbf P_i\mathbf H_{l,i}^{H}
+
\sigma_l^2\mathbf I_{N_l}
\right)
\mathbf w_l(\mathbf A_l)
}\\
&=
\frac{
\displaystyle
\sum_{n=1}^{N_l}\sum_{r=1}^{N_l}
[\mathbf w_l]_n^*(\mathbf A_l)
\left[
\mathbf H_{l,k}\mathbf P_k\mathbf H_{l,k}^{H}
\right]_{n,r}
[\mathbf w_l]_r(\mathbf A_l)
}{
\displaystyle
\sum_{n=1}^{N_l}\sum_{r=1}^{N_l}
[\mathbf w_l]_n^*(\mathbf A_l)
\left[
\sum_{i\neq k}
\mathbf H_{l,i}\mathbf P_i\mathbf H_{l,i}^{H}
+
\sigma_l^2\mathbf I_{N_l}
\right]_{n,r}
[\mathbf w_l]_r(\mathbf A_l)
},
\end{align}
with components computed as \eqref{eq:gamma_full_and_derivative}.
\begin{figure*}
\begin{equation}
\begin{aligned}
&[\mathbf w_l]_n(\mathbf A_l)=
\frac{1}{\sqrt{N_l}}\sqrt{\eta_{l,n}}
\sum_{k\in\mathcal U_l}
\sum_{m=1}^{M_k}
a_{l,k,m}
\frac{
1+\Re\left\{
e^{j\boldsymbol{\kappa}_{l,k,m}^{\top}\mathbf u_{l,n}} e^{-j \beta_s \rho_{l, n}}
\right\}
}{2}
e^{-\alpha_s\rho_{l,n}}
e^{-j \beta_s \rho_{l, n}},
\\
&\left[
\mathbf H_{l,k}\mathbf P_k\mathbf H_{l,k}^{H}
\right]_{n,r}
=
\beta_{lk}\cos\theta_{lk}
\sum_{m=1}^{M_k}
\sum_{t=1}^{M_k}
[\mathbf P_k]_{m,t}
e^{j\boldsymbol{\kappa}_{l,k,m}^{\top}\mathbf u_{l,n}} e^{-j\boldsymbol{\kappa}_{l,k,t}^{\top}\mathbf u_{l,r}}
 e^{-j\frac{2\pi}{\lambda}d_{l,k,m}^{(0)}} e^{+j\frac{2\pi}{\lambda}d_{l,k,t}^{(0)}} ,
\\
&\quad
\left[
\sum_{i\neq k}
\mathbf H_{l,i}\mathbf P_i\mathbf H_{l,i}^{H}
+
\sigma_l^2\mathbf I_{N_l}
\right]_{n,r}
=
\sum_{i\neq k}
\beta_{li}\cos\theta_{li}
\sum_{m=1}^{M_i}
\sum_{t=1}^{M_i}
[\mathbf P_i]_{m,t}
e^{j\boldsymbol{\kappa}_{l,i,m}^{\top}\mathbf u_{l,n}}
e^{-j\boldsymbol{\kappa}_{l,i,t}^{\top}\mathbf u_{l,r}}
e^{-j\frac{2\pi}{\lambda}d_{l,i,m}^{(0)}} e^{+j\frac{2\pi}{\lambda}d_{l,i,t}^{(0)}} 
+
\sigma_l^2
\begin{cases}
1, & n=r,\\
0, & n\neq r,
\end{cases}
\end{aligned}
\label{eq:gamma_full_and_derivative}
\end{equation}
\hrulefill
\end{figure*}
The  RHS amplitude optimization problem at AP \(l\) is
formulated as
\begin{align}
(\mathcal P{1}):
\max_{\mathbf A_l}
\quad
&
\sum_{k\in\mathcal U_l}
\gamma_{l,k}(\mathbf A_l)
\\
\mathrm{s.t.}\quad
&
a_{l,k,m}\geq 0,
k\in\mathcal U_l,m=1,\ldots,M_k,
\\
&
\sum_{k\in\mathcal U_l}
\sum_{m=1}^{M_k}
a_{l,k,m}
=
1 .
\label{eq:P0}
\end{align}
 $(\mathcal P{1})$ is a nonconvex sum-of-ratios problem.
\begin{proposition}[Projected-gradient update for \((\mathcal P{1})\)]
At iteration $q$, the  RHS coefficients are updated by
\begin{equation}
a_{l,k,m}^{(q+1)}
=
\left[
a_{l,k,m}^{(q)}
+
\mu^{(q)}
\sum_{k\in\mathcal U_l}
\left.
\frac{\partial \gamma_{l,k}}
{\partial a_{l,k,m}}
\right|_{\mathbf A_l=\mathbf A_l^{(q)}}
-
\nu_l^{(q)}
\right]^+ ,
\label{eq:pg_update}
\end{equation}
where $[x]^+=\max\{x,0\}$, $\mu^{(q)}>0$ is the step size,
and \(\nu_l^{(q)}\) is chosen such that
$
\sum_{k\in\mathcal U_l}
\sum_{m=1}^{M_k}
a_{l,k,m}^{(q+1)}=1.
$
This update is the Euclidean projection of the gradient-ascent
point onto the simplex feasible set of \((\mathcal P{1})\).
\end{proposition}

\begin{proof}
Define
$
\tilde a_{l,k,m}^{(q)}
=
a_{l,k,m}^{(q)}
+
\mu^{(q)}
\sum_{k\in\mathcal U_l}
\left.
\frac{\partial \gamma_{l,k}}
{\partial a_{l,k,m}}
\right|_{\mathbf A_l=\mathbf A_l^{(q)}} .
$
The projection onto the simplex solves
$
\min_{\mathbf A_l}
\sum_{k\in\mathcal U_l}
\sum_{m=1}^{M_k}
\left(
a_{l,k,m}-\tilde a_{l,k,m}^{(q)}
\right)^2
$
subject to
$
a_{l,k,m}\ge0,
$
$
\sum_{k\in\mathcal U_l}
\sum_{m=1}^{M_k}
a_{l,k,m}=1.
$
The KKT conditions are
$
a_{l,k,m}-\tilde a_{l,k,m}^{(q)}
+\nu_l^{(q)}-\lambda_{l,k,m}=0,
$
$
a_{l,k,m}\ge0,
$
$
\lambda_{l,k,m}\ge0,
$
$
\lambda_{l,k,m}a_{l,k,m}=0,
$
$
\sum_{k\in\mathcal U_l}
\sum_{m=1}^{M_k}
a_{l,k,m}=1.
$
Thus,
$
a_{l,k,m}^{(q+1)}
=
\left[
\tilde a_{l,k,m}^{(q)}-\nu_l^{(q)}
\right]^+,
$
which gives \eqref{eq:pg_update}.
\end{proof}
Specifically, $\frac{\partial \gamma_{l,k}}
{\partial a_{l,k,m}}$ is computed as 
\begin{equation}
\begin{aligned}
\frac{\partial \gamma_{l,k}}
{\partial a_{l,k,m}}
&=
\frac{
2\Re\!\left\{
\left(
\frac{\partial \mathbf w_l}
{\partial a_{l,k,m}}
\right)^H
\mathbf H_{l,k}\mathbf P_k\mathbf H_{l,k}^{H}
\mathbf w_l
\right\}
}{
\mathbf w_l^H
\left(
\sum_{i\neq k}
\mathbf H_{l,i}\mathbf P_i\mathbf H_{l,i}^{H}
+
\sigma_l^2\mathbf I_{N_l}
\right)
\mathbf w_l
}-
\\
&
\frac{
\mathbf w_l^H
\mathbf H_{l,k}\mathbf P_k\mathbf H_{l,k}^{H}
\mathbf w_l
}{
\left[
\mathbf w_l^H
\left(
\sum_{i\neq k}
\mathbf H_{l,i}\mathbf P_i\mathbf H_{l,i}^{H}
+
\sigma_l^2\mathbf I_{N_l}
\right)
\mathbf w_l
\right]^2
}\times
\\
&
2\Re\!\left\{
\left(
\frac{\partial \mathbf w_l}
{\partial a_{l,k,m}}
\right)^H
\left(
\sum_{i\neq k}
\mathbf H_{l,i}\mathbf P_i\mathbf H_{l,i}^{H}
+
\sigma_l^2\mathbf I_{N_l}
\right)
\mathbf w_l
\right\}.
\end{aligned}
\end{equation}

\begin{algorithm}[t]
\caption{Two-stage ranking-based AP association and RHS configuration per association interval}
\label{alg:proposed}
\begin{algorithmic}[1]
\REQUIRE
\(\{\mathbf q_l\}\), \(\{\mathbf Q_k\}\), \(\{\mathbf P_k\}\),
\(\{J_k\}\), \(\epsilon\), and \(I_{\max}\).
\ENSURE
Serving AP sets \(\{\mathcal S_k\}_{k=1}^{K}\), associated-UAV sets
\(\{\mathcal U_l\}_{l=1}^{L}\), and RHS weights \(\{\mathbf w_l\}_{l=1}^{L}\).

\STATE {Estimate $(\mathbf w_l^{(0)})^H\mathbf H_{l,k}$ and the local interference-plus-noise powers.}

\STATE \textbf{Pilot association phase:}
\FOR{each AP \(l=1,\ldots,L\)}
    \STATE Construct (or load from a codebook) the RHS basis patterns $\left\{\bar{b}_{l, k, m, n}\right\}$ based on the UAV angular directions, and form an initial probing weight $\mathbf{w}_l^{(0)}$ (e.g., a wide-beam setting) for pilot-based channel estimation.
\ENDFOR

\FOR{each UAV \(k=1,\ldots,K\)}
    \STATE Compute the preliminary score $\psi_{l,k}^{(0)}$ for all AP-UAV pairs
using the weight $\mathbf w_l^{(0)}$, before optimizing
$\mathbf A_l$.
    \STATE Select the \(J_k\) APs with the largest scores:
    $
    \mathcal S_k
    =
    \operatorname{Top}_{J_k}
    \left\{\psi_{1,k}^{(0)},\ldots,\psi_{L,k}^{(0)}\right\}.
    $
\ENDFOR

\FOR{each AP \(l=1,\ldots,L\)}
    \STATE Determine the associated UAV set
    $
    \mathcal U_l=\{k\mid l\in\mathcal S_k\}.
    $
    \IF{\(\mathcal U_l\neq\emptyset\)}
        \STATE Optimize \(\mathbf A_l\) over \(\mathcal U_l\) and
        construct the data-phase RHS weight \(\mathbf w_l\).
    \ELSE
        \STATE Keep AP \(l\) inactive in the data phase.
    \ENDIF
\ENDFOR

\STATE \textbf{Data transmission phase:}
\FOR{each active AP \(l\) with \(\mathcal U_l\neq\emptyset\)}
    \STATE Generate \(z_l=\mathbf w_l^H\mathbf y_l\) and forward \(z_l\)
    to the CPU.
\ENDFOR

\FOR{each UAV \(k=1,\ldots,K\)}
    \STATE Stack the selected AP outputs as
    $
    \mathbf z_{\mathcal S_k}=[z_l]_{l\in\mathcal S_k}.
    $
\ENDFOR

\RETURN
\(\{\mathcal S_k\}_{k=1}^{K}\), \(\{\mathcal U_l\}_{l=1}^{L}\),
and \(\{\mathbf w_l\}_{l=1}^{L}\).
\end{algorithmic}
\end{algorithm}

The above procedure is summarized in Algorithm 1 and is executed over association rounds. 
Within each interval, UAV positions and channels are regarded as quasi-static or predictable from recent pilots, and the obtained $\mathcal S_k$, $\mathcal U_l$, and $\mathbf w_l$ are used for the corresponding data-transmission phase. 
When UAVs move, the scores are updated and the AP association is recomputed.
To avoid circular dependence between AP association and RHS optimization, the CPU first determines $\{\mathcal S_k\}_{k=1}^{K}$ and $\mathcal U_l=\{k:l\in\mathcal S_k\}$ from the preliminary scores $\psi_{l,k}^{(0)}$. 
Then, $\mathbf A_l$ and $\mathbf w_l$ are optimized for the selected users and used only for final rate evaluation, not for AP re-ranking.

{The implementation analysis and computational-complexity assessment of  Algorithm 1 are provided as follows. After
estimating the post-RHS effective channels, AP $l$ computes its $K$
scores locally and reports only these real-valued scores to the CPU.
Thus, the CPU stores $LK$ scores. Sorting the $L$ candidate APs for
each UAV requires $\mathcal O(KL\log L)$ operations, which can be
reduced to $\mathcal O(KL)$ by partial Top-$J_k$ selection. The
subsequent RHS configurations are independently optimized at
different APs. In contrast, one-swap refinement evaluates up to
$KJ_k(L-J_k)$ candidate exchanges per sweep, whereas exhaustive joint
association requires $\binom{L}{J_k}^{K}$ candidates.}

\section{A Comparison with Existing AP Association Methods}
\label{sec:assoc_comparison}

The proposed AP association metric $\psi_{l,k}$ reflects the post-RHS scalar observation quality at AP $l$, including the desired signal power, residual inter-UAV leakage, and filtered noise power. 
This makes it a covariance-aware observation-quality score rather than a distance- or fading-only metric.

\subsubsection{Distance-, LSF-, and received-power-based methods}

Many user-centric cell-free schemes select APs according to distance, LSF, or received power, e.g.,
$
\mathcal S_k^{\mathrm{LSF}}
=
\operatorname{Top}_{J_k}
\{\beta_{1k},\ldots,\beta_{Lk}\},
$
and
$
\mathcal S_k^{\mathrm{dist}}
=
\operatorname{Top}_{J_k}^{\mathrm{min}}
\{d_{1k},\ldots,d_{Lk}\}.
$
Representative examples include strongest-channel selection \cite{Xue2025Handover}, power-based association \cite{Kandil2025Greedy}, and threshold-based activation \cite{Shi2023UAVThreshold}. 
These methods mainly rank APs by $\beta_{lk}$ or $d_{lk}^{-\alpha}$, whereas $\psi_{l,k}$ also depends on the  RHS response $\mathbf w_l$, the residual leakage
$
\mathbf w_l^{H}\mathbf H_{l,i}\mathbf P_i\mathbf H_{l,i}^{H}\mathbf w_l,
i\neq k,
$
and the filtered noise power
$
\sigma_l^2\mathbf w_l^{H}\mathbf w_l.
$
Thus, a nearby AP or an AP with strong LSF may still be unfavorable if its  post-RHS observation suffers from strong interference.

\subsubsection{Coverage-, load-, and pilot-based methods}

CAPA \cite{AlAlwani2025CAPA} and CLSF \cite{Liao2026EdgeEnhanced} improve scalability by considering coverage, pilot, or load constraints, but their AP contribution metrics are still largely based on large-scale fading. 
In contrast, $\psi_{l,k}$ penalizes APs whose  RHS receive response causes strong residual leakage. 
Therefore, it can be used as a replacement for $\beta_{lk}$ in coverage- or load-aware frameworks, e.g.,
$
\alpha_k^{\psi}
=
\frac{\sum_{l\in\mathcal S_k}\psi_{l,k}}
{\sum_{j=1}^{L}\psi_{j,k}}.
$

\subsubsection{MINLP- and learning-based methods}

MINLP- and learning-based methods jointly optimize AP association with beamforming, power control, or fronthaul constraints \cite{Liu2026EGAT,Zhou2025TDDPG,Du2025DMA}. 
They can directly target weighted sum rate or QoS, but often require binary optimization, iterative solvers, offline training, or policy learning. 
By contrast, after the  RHS receive-pattern configuration, the proposed method ranks APs using $\psi_{l,k}$. 
Under the diagonal or weak inter-AP disturbance-correlation approximation, the Top-$J_k$ rule provides a low-complexity solution for the adopted additive observation-quality surrogate.

\subsubsection{SIM/DMA-, mmWave-, and multi-criterion methods}

SIM-enabled \cite{Shi2025SIM}, DMA-based \cite{Du2025DMA}, and mmWave cell-free association methods often use distance, received power, effective channel gain, or weighted sum rate under their own hardware constraints. 
Multi-criterion fuzzy methods further combine indicators such as LSF, received power, and packet-loss probability \cite{Zeng2025FuzzyPacketLoss}. 
These metrics are useful in their respective settings, but they do not explicitly model the single-feed amplitude-controlled RHS receive response. 
The proposed score instead follows the single-feed RHS receive model through $\mathbf w_l$ and evaluates the  post-RHS signal-to-disturbance quality score. 
If reliability indicators are also required, they can be incorporated multiplicatively, e.g.,
$
\tilde{\psi}_{l,k}
=
(1-p_{l,k}^{\mathrm{loss}})\psi_{l,k},
$ where $p_{l,k}^{\mathrm{loss}}$ denotes the packet loss ratio.

\subsubsection{Conventional SINR-based association methods in phased-array MIMO systems}

For a conventional digitally combined AP, a user-specific output
$z_{l,k}=\mathbf v_{l,k}^{H}\mathbf y_l$ can be
formed, with
$
\gamma_{l,k}
=
\frac{
\mathbf v_{l,k}^{H}\mathbf H_{l,k}\mathbf P_k
\mathbf H_{l,k}^{H}\mathbf v_{l,k}
}{
\mathbf v_{l,k}^{H}
\left(
\sum_{i\neq k}\mathbf H_{l,i}\mathbf P_i\mathbf H_{l,i}^{H}
+\sigma_l^2\mathbf I_{N_l}
\right)
\mathbf v_{l,k}
}.
$
In contrast, each single-feed RHS AP considered here produces one
shared scalar output $z_l=\mathbf w_l^{H}\mathbf y_l$ with $\mathbf{w}_l^H=\left(\mathbf{v}_l^*\right)^H=\mathbf{v}_l^{\top}$. Its
probing-stage score is
$
\psi_{l,k}^{(0)}
=
\frac{
(\mathbf w_l^{(0)})^{H}\mathbf H_{l,k}\mathbf P_k
\mathbf H_{l,k}^{H}\mathbf w_l^{(0)}
}{
(\mathbf w_l^{(0)})^{H}
\left(
\sum_{i\neq k}\mathbf H_{l,i}\mathbf P_i\mathbf H_{l,i}^{H}
+\sigma_l^2\mathbf I_{N_l}
\right)
\mathbf w_l^{(0)}
}.
$
The two ratios are algebraically identical for
$\mathbf v_{l,k}=\mathbf w_l^{(0)}$; hence, the ratio itself is not
claimed as new. Their structural difference is that
$\mathbf v_{l,k}$ is generally user-specific, whereas
$\mathbf w_l$ is shared by the associated UAVs and constrained by
the RHS hardware.

{
Table~\ref{tab:association_comparison} compares the AP metrics and
association strategies of representative methods. Existing
low-complexity schemes mainly use large-scale information or
user-specific SINR, whereas optimization and learning methods incur
higher complexity. The proposed method uses post-RHS observation
quality under a shared probing pattern and yields a low-complexity
Top-$J_k$ solution for the lower-bound surrogate.
}
\begin{table}[t]
\centering
\caption{Comparison of representative AP association methods.}
\label{tab:association_comparison}
\renewcommand{\arraystretch}{1.10}
\setlength{\tabcolsep}{2.5pt}
\scriptsize

\begin{tabularx}{\columnwidth}{
|
>{\raggedright\arraybackslash}p{2.20cm}
>{\raggedright\arraybackslash}X
>{\raggedright\arraybackslash}p{3.10cm}
|}
\hline
Method
&
AP metric
&
Association strategy
\\
\hline

Distance/LSF/power
\cite{Xue2025Handover,Kandil2025Greedy,
Shi2023UAVThreshold}
&
Distance, LSF, or power
&
Ranking or thresholding
\\

Coverage/load/pilot
\cite{AlAlwani2025CAPA,Liao2026EdgeEnhanced}
&
LSF with network constraints
&
Constrained selection
\\

MINLP/learning
\cite{Liu2026EGAT,Zhou2025TDDPG,Du2025DMA}
&
Rate, QoS, or system state
&
Optimization or learning
\\

SIM/DMA/mmWave/
multi-criterion
\cite{Shi2025SIM,Du2025DMA,Wang2022mmWaveCF,
Zeng2025FuzzyPacketLoss}
&
Effective gain, rate, or multiple indicators
&
Architecture-specific design
\\

Conventional SINR
&
SINR with a user-specific combiner
&
SINR ranking
\\

Proposed
&
Post-RHS desired gain, leakage, and noise
&
Lower-bound Top-K ranking
\\
\hline
\end{tabularx}

\end{table}

\section{Analytical Findings and Proofs}

\subsection{A nearer AP is not necessarily better}

Consider UAV $k$ and two candidate APs $l$ and $m$. 
In conventional distance-based association, if $d_{lk}<d_{mk}$, AP $l$ would usually be preferred. 
However, under the RHS-enabled uplink reception model, the AP association priority is not determined only by distance or large-scale fading. 
It also depends on the  RHS-induced desired effective channel, the residual inter-UAV leakage after RHS reception, and the filtered noise power.

For AP $l$ and UAV $k$, define
$
\mathcal D_{l,k}
=
\mathbf w_l^{H}
\mathbf H_{l,k}\mathbf P_k\mathbf H_{l,k}^{H}
\mathbf w_l,
$
and
$
\mathcal I_{l,k}
=
\mathbf w_l^{H}
\left(
\sum_{i\neq k}
\mathbf H_{l,i}\mathbf P_i\mathbf H_{l,i}^{H}
+
\sigma_l^2\mathbf I_{N_l}
\right)
\mathbf w_l.
$
Then the per-AP association score can be written as
$
\psi_{l,k}
=
\mathcal D_{l,k}/\mathcal I_{l,k}.
$

\begin{corollary}[Observation-quality order reversal condition]
AP $l$ should be preferred over AP $m$ under the additive surrogate whenever
$
\psi_{l,k}>\psi_{m,k}.
$
Equivalently, even if AP $l$ has a weaker desired effective signal power than AP $m$, namely
$
\mathcal D_{l,k}<\mathcal D_{m,k},
$
AP $l$ can still be preferred if its normalized association score is larger:
$
\mathcal D_{l,k}/\mathcal I_{l,k}
>
\mathcal D_{m,k}/\mathcal I_{m,k}.
$
This can happen when AP $l$ has a sufficiently smaller residual disturbance power $\mathcal I_{l,k}$.
\end{corollary}

\begin{proof}
Please refer to Appendix A.
\end{proof}

\subsection{Large scale height-angle tradeoff}
From the element-level channel model in \eqref{eq:H_lk_target_uav}, the
large scale geometry-dependent gain of the AP-UAV link is 
$
\beta_{lk}\cos\theta_{lk}
$.
Using
$
\beta_{lk}=\bar C d_{lk}^{-\alpha}\chi_{lk}
$,
$
d_{lk}=\sqrt{r_{lk}^{2}+z_{lk}^{2}}
$,
and
$
\cos\theta_{lk}=z_{lk}/d_{lk}
$,
we have
\begin{equation}
\beta_{lk}\cos\theta_{lk}
=
\bar C\chi_{lk}
z_{lk}
\left(
r_{lk}^{2}+z_{lk}^{2}
\right)^{-\frac{\alpha+1}{2}} .
\label{eq:geometry_gain}
\end{equation}

To isolate the height-angle tradeoff, we consider a local comparison
where the transmit covariance, RHS combining pattern, and shadowing term
are fixed. Then the desired received power is proportional to the
geometry-dependent factor in \eqref{eq:geometry_gain}. For a fixed
horizontal distance $r_{lk}=r$, define
\begin{equation}
f(z_{lk})
=
z_{lk}
\left(
r^{2}+z_{lk}^{2}
\right)^{-\frac{\alpha+1}{2}} .
\label{eq:height_gain_function}
\end{equation}
\begin{corollary}[Height-angle tradeoff]
For any $\alpha>0$ and fixed $r_{lk}=r>0$, the function
$f(z_{lk})$ in \eqref{eq:height_gain_function} attains a unique
maximum at
\begin{equation}
z_{lk}^{\star}
=
\frac{r}{\sqrt{\alpha}} .
\end{equation}
Moreover,
\begin{equation}
\frac{\partial f}{\partial z_{lk}}
=
\left(r^{2}+z_{lk}^{2}\right)^{-\frac{\alpha+3}{2}}
\left(r^{2}-\alpha z_{lk}^{2}\right).
\end{equation}
Therefore, $f(z_{lk})$ increases for
$0<z_{lk}<z_{lk}^{\star}$ and decreases for
$z_{lk}>z_{lk}^{\star}$.
\end{corollary}
\begin{proof}
Please refer to Appendix B.
\end{proof}
Note that the straight-line, approximately constant-altitude trajectory is a
mission-prescribed model for aisle and shelf inspection, rather than
a requirement of the proposed association rule. If UAV positions were
controllable, trajectory optimization could improve the desired gain
by reducing $r_{lk}$ or moving $z_{lk}$ toward
$r_{lk}/\sqrt{\alpha}$, and could reduce inter-UAV leakage through
greater angular separation. For curved or altitude-varying
trajectories, the proposed method remains applicable by updating
$\mathbf H_{l,k}(t)$ and $\psi_{l,k}(t)$ at each association interval.

\subsection{More APs do not always bring significant gains}

According to the association score, let the candidate APs for UAV $k$ be ordered as
$
\psi_{(1)k}
\geq
\psi_{(2)k}
\geq
\cdots
\geq
\psi_{(L)k}.
$
Define
$
\gamma_{(j)k}
=
\psi_{(j)k}.
$
For the additive scalar surrogate used by the low-complexity ranking rule, if the first $J$ APs are selected, the surrogate rate is
$
R_k(J)
=
\log_2
\left(
1+
\sum_{j=1}^{J}
\gamma_{(j)k}
\right).
$
The marginal gain of adding the $J$-th AP is
\begin{equation}
\Delta R_k(J)
=
R_k(J)-R_k(J-1)
=
\log_2
\left(
1+
\frac{
\gamma_{(J)k}
}{
1+\sum_{j=1}^{J-1}\gamma_{(j)k}
}
\right).
\label{eq:deltaR_revised}
\end{equation}

\begin{corollary}[Cluster saturation law for the additive surrogate]
For any $\varepsilon>0$, if
$
\gamma_{(J)k}
\leq
\varepsilon
\left(
1+
\sum_{j=1}^{J-1}
\gamma_{(j)k}
\right),
$
then
$
\Delta R_k(J)
\leq
\log_2(1+\varepsilon).
$
\end{corollary}

\begin{proof}
Please refer to Appendix C.
\end{proof}

\subsection{Design guidelines for indoor industrial AP deployment}
Conventional distance- or fading-based association fails in indoor industrial UAV networks for two reasons. First, when UAVs are clustered along similar inspection paths, they present nearly identical $\beta_{l k}$ and ranges to a given AP; ranking by these scalars cannot distinguish good APs from those that suffer severe inter-UAV overlap. The proposed score $\psi_{l, k}$ directly senses this interference in its denominator, penalizing APs that cannot spatially separate the target UAV. Second, dense blockages create strong but angularly mismatched links. An AP with a large $\beta_{l k}$ may collect multipath energy that the RHS receive pattern cannot effectively combine. Because $\psi_{l, k}$ evaluates the post-RHS signal quality, it naturally deprioritizes APs whose effective receive gain is degraded by clutter. In open, uncluttered settings with randomly distributed UAVs, these advantages largely disappear and simpler association suffices.
The system model closely reflects the geometry and operation of an indoor industrial hall. The AP-UAV geometry $\left(r_{l k}, z_{l k}, \theta_{l k}\right)$ is fundamental to the height-angle tradeoff (Corollary 2); an optimal vertical separation exists only because APs are ceiling-mounted and UAVs operate at a fixed altitude. The configuration is characteristic of factories and warehouses but absent in general deployments. The UAV motion is modeled as deterministic, straight-line, constant-speed, constant-altitude flight, which matches inspection tasks along aisles and shelves while enabling trajectory-predictable AP association updates in Algorithm 1. Finally, the channel model employs a large path-loss exponent and a strong log-normal shadowing variance that capture the severe, heterogeneous indoor environment; it is this large variance that produces the distance-order reversal probability reported in Corollary 1, confirming the unreliability of geometry-based association in heavily obstructed industrial spaces.

The above analysis leads to the following deployment guidelines.

\textbf{Rule 1:} AP association should use the post-RHS signal-to-disturbance score $\psi_{l,k}$ rather than large-scale fading $\beta_{lk}$ alone.

\textbf{Rule 2:} An AP should preferably serve UAVs with sufficiently distinct transverse wavevectors or angular directions to reduce post-RHS leakage.

\textbf{Rule 3:} Increasing the serving-cluster size gives diminishing returns, so a small set of high-quality APs is often sufficient.

\textbf{Rule 4:} The installation height of ceiling-mounted RHS APs should jointly consider path loss and projection gain, rather than only minimizing distance.

\textbf{Rule 5:} When UAVs are sparse and inter-UAV overlap is weak, conventional large scale fading-based or distance-based clustering can still serve as a low-complexity scheme.

\section{Simulation}

\subsection{Simulation setup}

\begin{table}[h]
\centering
\caption{Simulation parameters}
\label{tb:para}
\begin{tabular}{|p{1cm}p{6.8cm}|}
\hline
Parameter & Value \\
\hline
hall 
& $D_x^{\rm hall}$ = $40$ m, $D_y^{\rm hall}$ = $30$ m, $D_z^{\rm hall}$=$12$ m \\
$L,K,J$ 
& $L=16$, $K=24$, $J=4$ \\

$M_k,N_l$ 
& $M_k=4,\ \forall k$; $N_l=256,\ \forall l$ \\

$\lambda,\mathbf k_s$ 
& $\lambda=0.025$ m, $2\pi/\lambda=251.33$ rad/m, $\mathbf k_s=[251.33,0,0]^\top$ \\

$\mathbf q_l$ 
& Ceiling APs on a $4\times4$ grid: $q_{l,x}\in\{5,15,25,35\}$ m, $q_{l,y}\in\{3.75,11.25,18.75,26.25\}$ m, $q_{l,z}=12$ m \\

$\mathbf u_{l,n}$ 
& Centered $16\times16$ uniform planar grid, element spacing $0.005$ m;  \\

$\mathbf q_{k,1}$ 
& Users $1$-$8$: $(5,5,1.5)+[\epsilon_x,\epsilon_y,0]$ m;
users $9$-$16$: $(18,13,1.5)+[\epsilon_x,\epsilon_y,0]$ m;
users $17$-$24$: $(32,22,1.5)+[\epsilon_x,\epsilon_y,0]$ m,
where $\epsilon_x,\epsilon_y\sim\mathrm{uniform}[-1.5,1.5]$ m  \\

$\bar{\mathbf q}_k$ 
& $\mathbf q_{k,1}+[0.075,0,0]^\top$ \\

$\mathbf Q_k$ 
& $[\mathbf q_{k,1},\mathbf q_{k,1}+[0.05,0,0]^\top,\mathbf q_{k,1}+[0.10,0,0]^\top,\mathbf q_{k,1}+[0.15,0,0]^\top]$ \\

$\mathbf d_k,v_k$ 
& $\mathbf d_k=[0,1,0]^\top$, $v_k=1$ m/s, $\forall k$ \\

$\alpha,\chi_{lk},\xi_{lk}$ 
& $\alpha=2.2$, $\chi_{lk}=10^{\xi_{lk}/10}$, $\xi_{lk}\sim\mathcal N(0,4^2)$ \\

$\eta_{l,n},\alpha_s$ 
& $\eta_{l,n}=1$; $\alpha_s=0$ unless otherwise stated;  \\

$\mathbf P_k,p_k$ 
& $\mathbf P_k=0.1\mathbf I_{M_k}$, $p_k=0.1,\ \forall k$ \\

$\sigma_l^2$ 
& $4.0\times10^{-13},\ \forall l$ \\
\hline
\end{tabular}
\end{table}

We consider the uplink UAV-centric cell-free scenario, where each single-feed RHS-enabled AP forms one  scalar output
$
z_l=\mathbf w_l^H\mathbf y_l
$
in each resource block. 
The same  scalar output can be forwarded to the CPU and used for detecting different UAVs. 
For UAV $k$, the CPU stacks the selected AP outputs $\{z_l:l\in\mathcal S_k\}$ and evaluates the achievable rate using \eqref{eq:rate_logdet_expanded}.
The proposed association rule ranks candidate APs according to the  post-RHS score. 

For comparison, the following benchmark association rules are considered:
\begin{enumerate}
\item \textit{Proposed score-based association}: AP association based on the post-RHS score $\psi_{l,k}$, which jointly accounts for the link strength $\beta_{lk}$, the  RHS response $\mathbf w_l$, residual inter-UAV leakage, and the local filtered noise term.
\item \textit{Large-scale fading}: AP association based only on $\beta_{lk}$.
\item \textit{Nearest AP}: AP association based only on the distance $d_{lk}$.
\item \textit{Pathloss and projection}: AP association based on $d_{lk}^{-\alpha}\cos\theta_{lk}$, which accounts for path loss and the receiving angle factor but not residual interference.
\item \textit{CAPA}: coverage-aware AP association \cite{AlAlwani2025CAPA}.
\item \textit{CLSF}: load-constrained large-scale fading \cite{Liao2026EdgeEnhanced}.
\end{enumerate}

The considered performance metrics include:
\begin{itemize}
\item the minimum UAV spectral efficiency,
$
\min_{1\le k\le K}
R_k;
$

\item the average spectral efficiency,
$
\frac{1}{K}
\sum_{k=1}^{K}
R_k;
$

\item Jain's fairness index,
$
\mathcal J
=
\frac{\left(\sum_{k=1}^{K}R_k\right)^2}
{K\sum_{k=1}^{K}R_k^2};
$

\item the transmit-power-normalized energy efficiency,
$
\frac{\sum_{k=1}^{K}R_k}
{\sum_{k=1}^{K}\operatorname{tr}(\mathbf P_k)};
$

\item the normalized single-AP score,
$
\frac{\psi_{l,k}(z)}
{\max_z \psi_{l,k}(z)};
$

\item the average marginal gain of adding APs to the serving cluster, as defined in \eqref{eq:deltaR_revised};

\item the probability of distance-order reversal,
$
\frac{
\sum_{k=1}^{K}
\sum_{l=1}^{L}
\sum_{\substack{m=1\\ m\neq l}}^{L}
\mathbb{I}
\left(
d_{lk}<d_{mk},\ 
\psi_{l,k}<\psi_{m,k}
\right)
}{
\sum_{k=1}^{K}
\sum_{l=1}^{L}
\sum_{\substack{m=1\\ m\neq l}}^{L}
\mathbb{I}
\left(
d_{lk}<d_{mk}
\right)
}
$;

\item the Top-$J$ overlap ratio;

\end{itemize}

The simulation scenario and default parameter values are summarized in Table \ref{tb:para}. The UAVs' initial positions are clustered into three groups to emulate localized warehouse inspection task regions.  All UAVs move along straight-line trajectories at a constant altitude with a constant speed of $1~\mathrm{m/s}$ and a common heading direction $\mathbf d=[0,1,0]^{\top}$. The flight duration is $4$ seconds.

{To evaluate the performance-loss trade-off related to exhaustive search or optimization/learning-based methods, we further consider a small-scale system with $L=5, K=4, J=2, M_k=2$, and $N_l=36$; each marker represents one independent channel realization, and all methods use the same channels, RHS constraints, and CPU-side rate evaluator. The optimization/learning-based method uses one-swap refinement, which starts from the proposed Top-$J_k$ association and repeatedly replaces one selected AP with one unselected AP whenever the exchange improves the exact CPU-side sum rate.}

{To isolate the influence of industrial UAV clustering, we vary the
standard spatial spread of the UAVs around three localized inspection
regions from $0.5$ to $6$ m. The other parameters follow the default
setup with $L=16$, $K=24$, $J=4$, $M_k=4$, and $N_l=256$. All
association methods use the same channel realizations.}

\subsection{Simulation results and discussion}

Fig.~\ref{fig:min_rate_K} shows the minimum UAV data rate versus the number of  UAVs. As the number of requesting UAVs increases, the minimum rate decreases for all association schemes because the uplink becomes more interference-limited and the weakest UAVs suffer from stronger inter-UAV overlap. The proposed  method achieves the highest minimum data rate over the whole range of $K$, demonstrating its advantage in protecting weak UAVs. This is because the proposed score considers not only the desired effective RHS gain but also the residual overlap and local disturbance covariance. Among the benchmarks, the nearest-AP rule gives the lowest rate, which confirms that geometric proximity alone is insufficient for RHS-enabled indoor UAV cell-free reception.

\begin{figure}[h]
\centering
\includegraphics[width=0.4\textwidth]{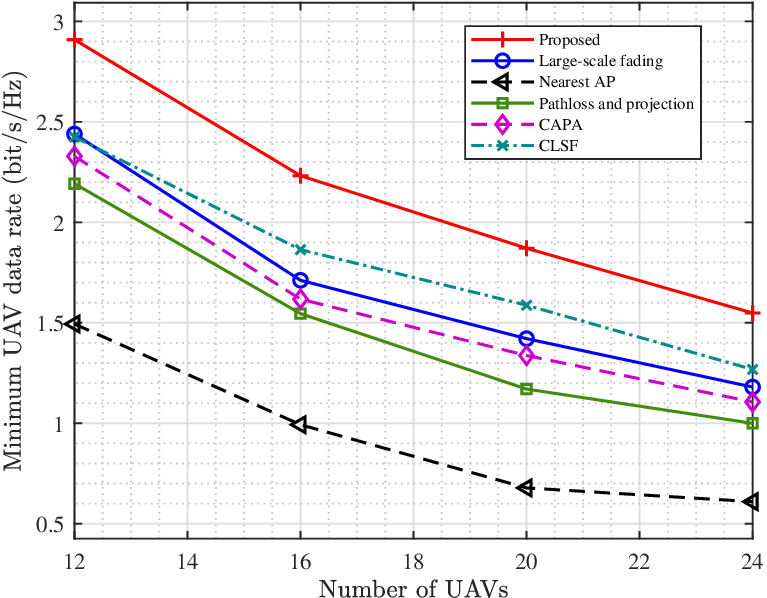}
\caption{Minimum UAV data rate versus the number of  UAVs.}
\label{fig:min_rate_K}
\end{figure}

Fig.~\ref{fig:avg_se_K} plots the average spectral efficiency versus the number of  UAVs. The average spectral efficiency decreases with $K$ because more simultaneous UAV transmissions increase the aggregate interference and reduce the effective post-combining SINR. The proposed  association consistently outperforms all benchmark methods. This indicates that post-RHS signal-to-disturbance score ranking improves not only the worst-user performance but also the overall system throughput. The performance gap between the proposed method and the nearest-AP rule is significant, showing that selecting APs according to distance may lead to poor AP choices when the nearest AP has unfavorable RHS angular response or strong multi-UAV overlap.

\begin{figure}[h]
\centering
\includegraphics[width=0.4\textwidth]{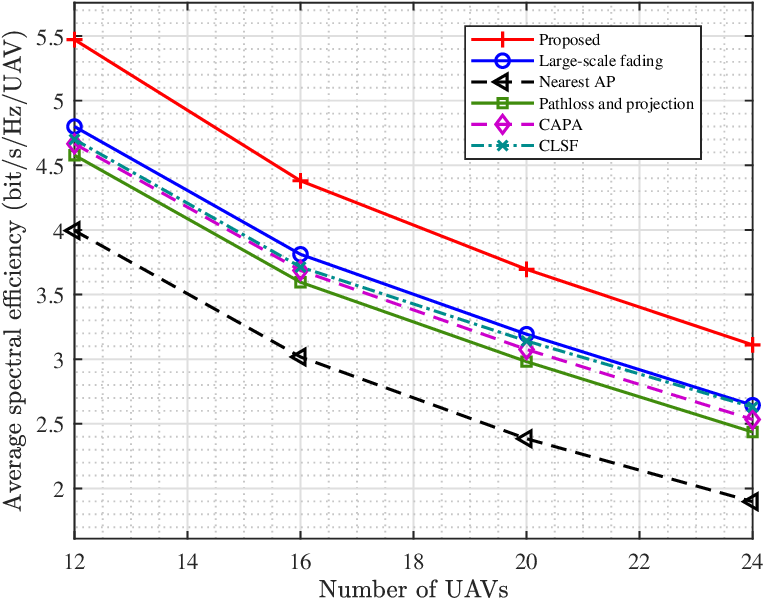}
\caption{Average spectral efficiency versus the number of  UAVs.}
\label{fig:avg_se_K}
\end{figure}

Fig.~\ref{fig:fairness_J} shows Jain's fairness index versus the number of serving APs $J$. The fairness index increases as $J$ becomes larger, because adding more serving APs provides additional macro-diversity and reduces the probability that a UAV is served only by weak or highly interfered APs. The proposed method achieves the best fairness performance among the compared schemes. This confirms that the proposed association score is effective in balancing service quality across UAVs. The fairness gain becomes more visible when the serving-cluster size is moderate or large, where the proposed method can select APs that jointly provide strong desired gain and low disturbance.

\begin{figure}[h]
\centering
\includegraphics[width=0.4\textwidth]{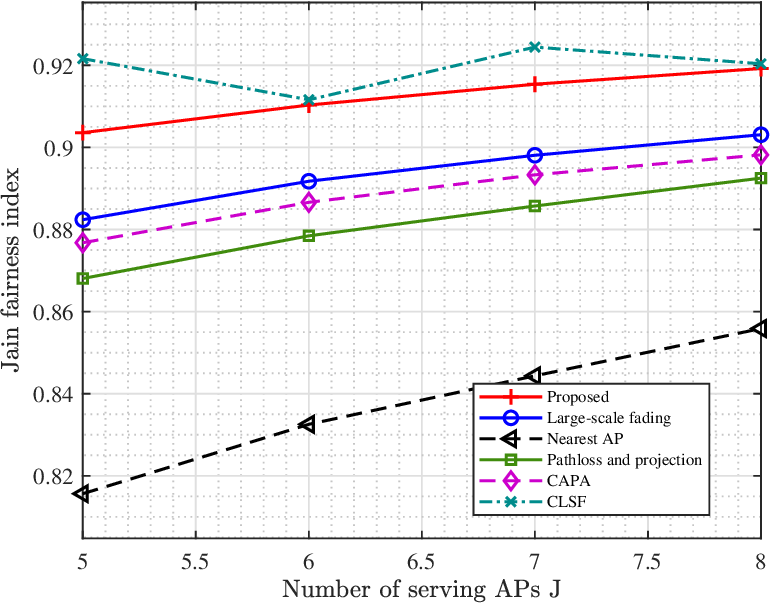}
\caption{Jain fairness index versus the number of serving APs.}
\label{fig:fairness_J}
\end{figure}

Fig.~\ref{fig:ee_K} compares the energy efficiency of different AP association schemes. The energy efficiency decreases as the number of requesting UAVs increases, since the total transmit power grows with $K$ while the achievable sum spectral efficiency is increasingly limited by inter-UAV overlap. The proposed  method achieves the highest energy efficiency in all simulated cases. This means that, for the same transmit-power budget, the proposed ranking-based AP association can extract more useful information from the distributed RHS-enabled APs. The nearest-AP rule again gives the lowest energy efficiency because it may select APs with high disturbance even if they are close to the UAV.

\begin{figure}[h]
\centering
\includegraphics[width=0.4\textwidth]{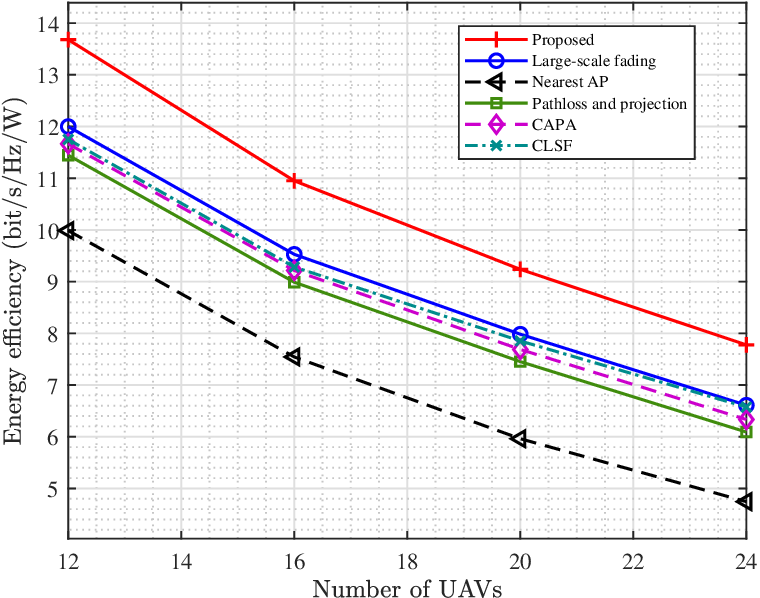}
\caption{Energy efficiency versus the number of  UAVs.}
\label{fig:ee_K}
\end{figure}

Fig.~\ref{fig:cdf_rate} presents the CDF of the per-UAV data rate. The CDF curve of the proposed  method is shifted to the right compared with all benchmark schemes, which means that it improves the rate distribution of individual UAVs. In particular, for the same outage probability, the proposed method provides a higher per-UAV data rate. This lower-tail improvement is important for indoor industrial UAV applications, where reliable service to weak UAVs is often more critical than peak throughput. The nearest-AP based scheme has the leftmost CDF, confirming that distance-based association leads to inferior rate reliability in the considered RHS-enabled cell-free network.

\begin{figure}[h]
\centering
\includegraphics[width=0.4\textwidth]{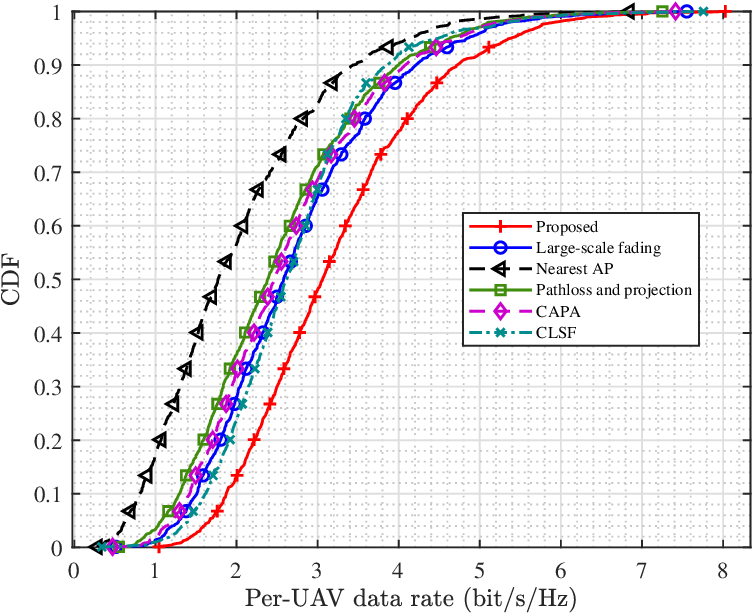}
\caption{CDF of the per-UAV data rate.}
\label{fig:cdf_rate}
\end{figure}

Fig.~\ref{fig:reversal} verifies the first analytical finding that the nearest AP is not necessarily the best AP. The distance-order reversal probability is clearly nonzero for all considered shadowing standard deviations and UAV numbers. This means that a farther AP can provide a better association score than a nearer AP. The reason is that the proposed score depends on the effective RHS channel gain, angular separability, inter-UAV overlap, and disturbance covariance, rather than only on distance. Moreover, the reversal probability becomes higher when the shadowing standard deviation increases, which indicates that the distance order becomes less reliable in more heterogeneous indoor propagation environments.

\begin{figure}[h]
\centering
\includegraphics[width=0.4\textwidth]{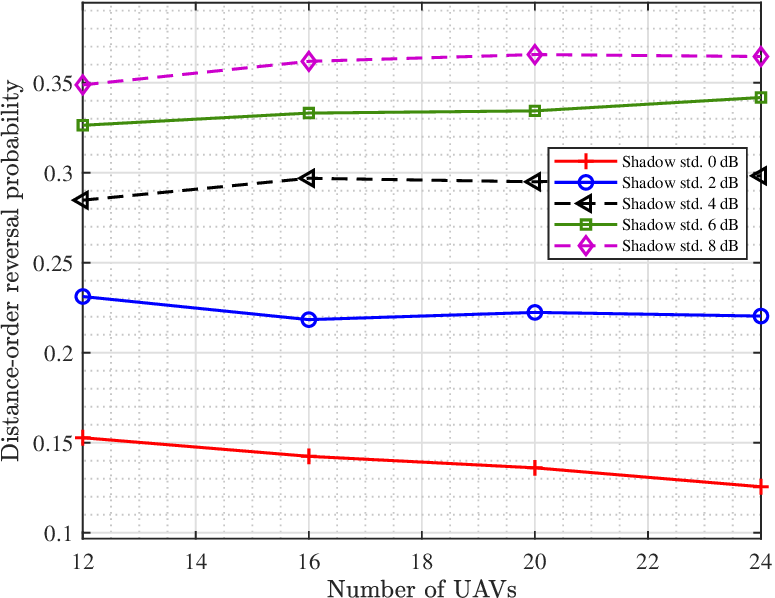}
\caption{Distance-order reversal probability versus the number of  UAVs.}
\label{fig:reversal}
\end{figure}

Fig.~\ref{fig:height_tradeoff} verifies the second analytical finding, namely the AP-UAV height-angle tradeoff. For each horizontal distance $r$, the normalized single-AP score first increases and then decreases with the AP-UAV relative height difference $z$. This confirms the existence of an interior optimal height difference. When $z$ is too small, the UAV observes the ceiling-mounted RHS with a large incidence angle, which reduces the projection gain. When $z$ is too large, the increased propagation distance and path loss dominate. Therefore, AP deployment height should not be determined only by the shortest-distance rule, but should jointly consider path loss and RHS angular response.

\begin{figure}[h]
\centering
\includegraphics[width=0.4\textwidth]{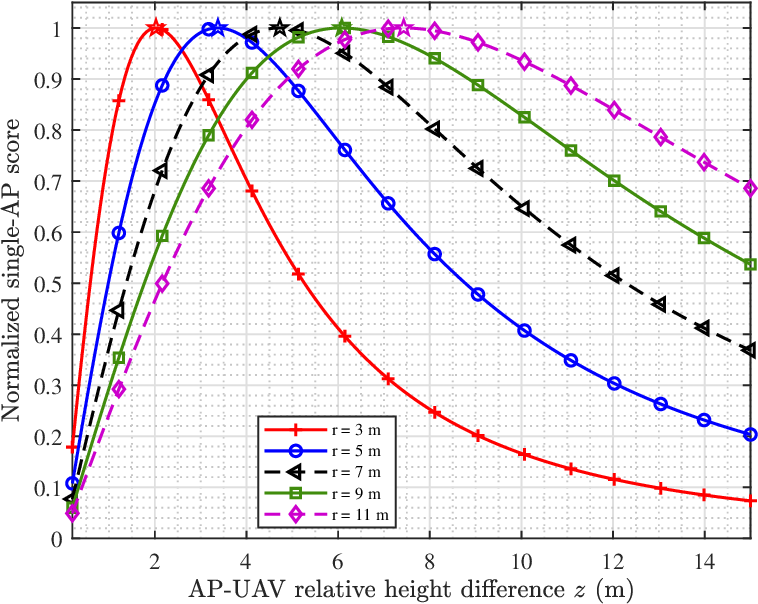}
\caption{Normalized single-AP score versus AP-UAV relative height difference.}
\label{fig:height_tradeoff}
\end{figure}

Fig.~\ref{fig:saturation} verifies the third analytical finding, i.e., the cluster saturation law. The average marginal gain $\Delta R_k(J)$ decreases as the number of serving APs increases. This shows that the first few selected APs provide most of the useful rate improvement, while adding more APs brings only a small additional gain. The result supports the proposed Top-$J$ ranking rule and suggests that a practical system can use an adaptive stopping criterion: if the marginal gain of adding one more AP is below a prescribed threshold, the serving cluster can stop expanding. 

\begin{figure}[h]
\centering
\includegraphics[width=0.4\textwidth]{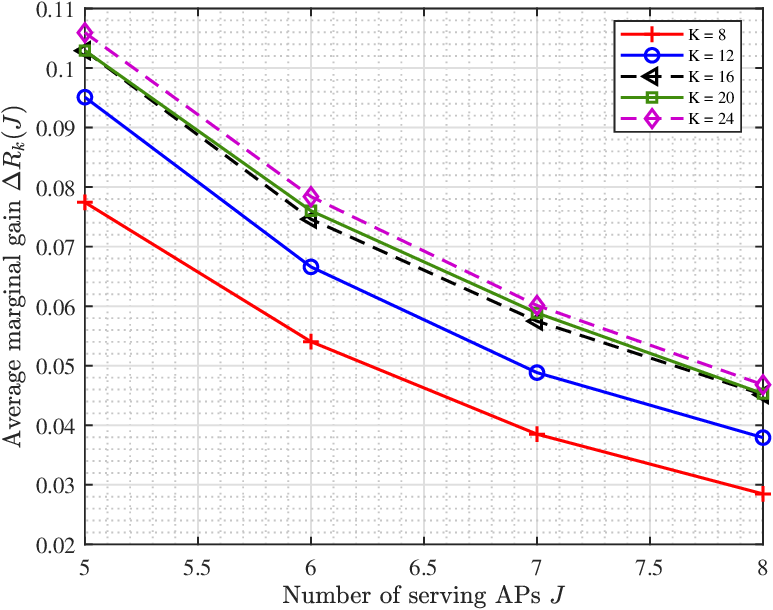}
\caption{Average marginal gain versus the number of serving APs.}
\label{fig:saturation}
\end{figure}

Fig.~\ref{fig:cdf_overlap} shows the CDF of the selected AP overlap ratio. A smaller selected AP overlap ratio indicates that the chosen APs experience less residual inter-UAV leakage relative to the useful desired component. The proposed  curve is shifted to the left, which means that the proposed method selects APs with lower effective overlap than the benchmark schemes. This behavior directly follows from the disturbance term in the proposed association score. In contrast, large-scale-fading, pathloss-and-projection, and nearest-AP based methods may select APs with strong average channels but severe multi-UAV overlap. Therefore, this result confirms the key motivation of the proposed method: AP association in RHS-enabled indoor UAV cell-free networks should be based on post-combining quality rather than distance or large-scale fading alone.

\begin{figure}[h]
\centering
\includegraphics[width=0.4\textwidth]{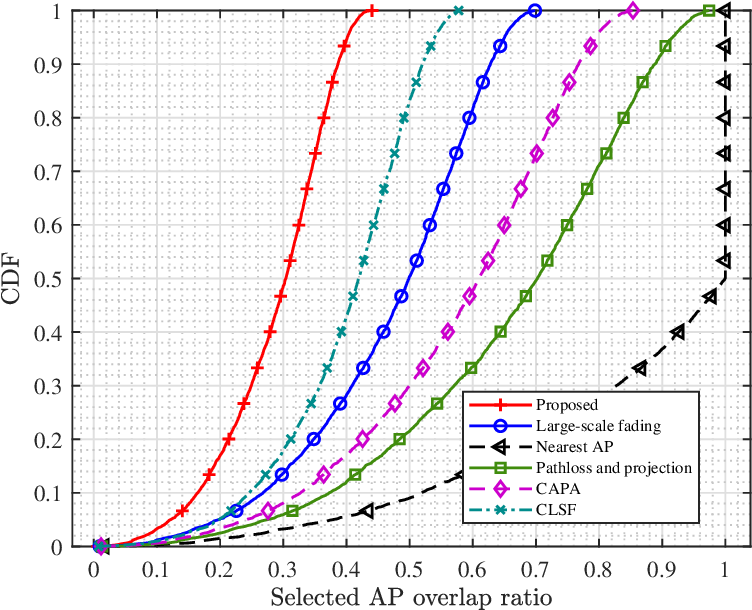}
\caption{CDF of the selected AP overlap ratio.}
\label{fig:cdf_overlap}
\end{figure}

{Fig.~\ref{fig:tradeoff} shows a clear performance-complexity tradeoff. The proposed Top-$J_k$ method requires only about $10^{-4}-10^{-3} \mathrm{~s}$, but achieves a lower average spectral efficiency. The one-swap optimization improves the spectral efficiency and approaches the exhaustive-search result while requiring only about $10^{-2} \mathrm{~s}$. Exhaustive search provides the highest-performance reference, but its runtime is around $10^0 \mathrm{~s}$. Thus, the proposed method offers the lowest online complexity while only about 14\% performance loss.}

\begin{figure}[h]
\centering
\includegraphics[width=0.4\textwidth]{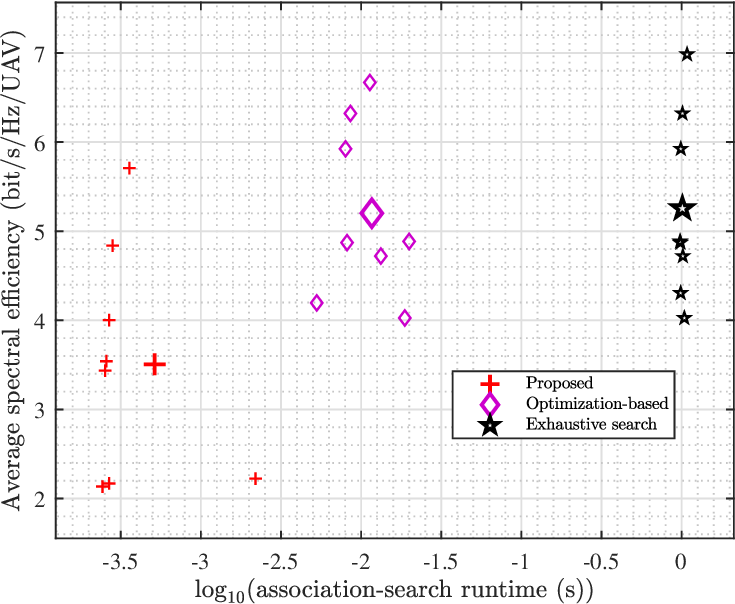}
\caption{{Performance-runtime tradeoff ($K=4$)}}
\label{fig:tradeoff}
\end{figure}

{Fig.~\ref{fig:cluster} evaluates the effect of localized industrial UAV operation.
When the UAV cluster spread is small, candidate APs observe more
similar UAV directions and hence stronger residual overlap. The
proposed method consistently selects APs with lower overlap than LSF
and CLSF because this interference is explicitly included in its
score. As the UAVs become more dispersed, the overlap decreases for
all methods. This result clarifies that the main scenario-specific
advantage of the proposed rule arises in clustered industrial
inspection deployments, rather than in open and sparsely populated
environments.}

\begin{figure}[h]
\centering
\includegraphics[width=0.4\textwidth]{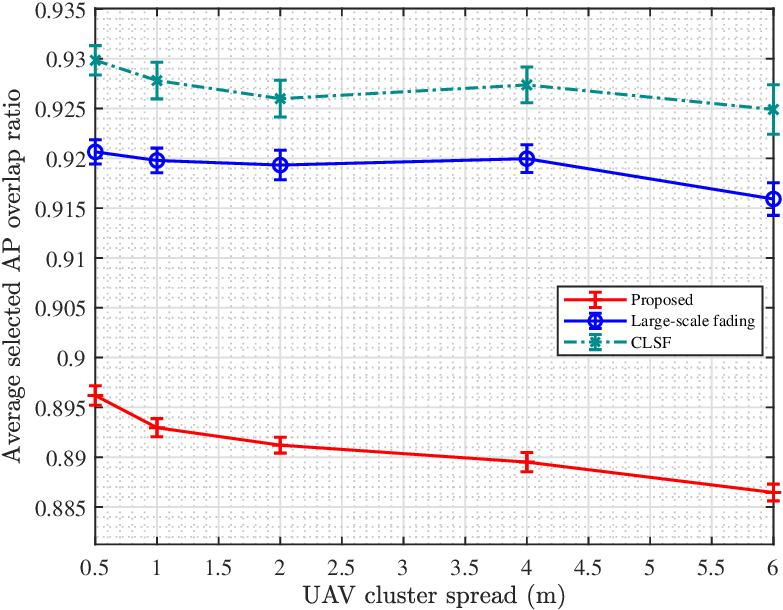}
\caption{{Average selected AP overlap ratio versus the UAV cluster
spread.}}
\label{fig:cluster}
\end{figure}

\section{Conclusion}
This paper investigated AP association for uplink RHS-enabled cell-free networks in indoor industrial UAV scenarios. A post-RHS signal-to-disturbance score was proposed as the association metric, which jointly captures large-scale fading, effective post-RHS channel gain, residual inter-UAV overlap after holographic combining, and local disturbance power. Under the assumption of weak inter-AP disturbance covariance, a low-complexity ranking rule based on per-AP scores was derived. The analysis established three findings: a nearer AP is not necessarily the optimal serving AP; the relative AP-UAV height difference creates a height-angle tradeoff with an interior optimum; and enlarging the serving cluster brings diminishing marginal gains. Simulations with clustered UAVs moving along straight-line constant-altitude inspection trajectories demonstrated that the proposed method improves the minimum UAV data rate, average spectral efficiency, fairness, and transmit-power-normalized energy efficiency compared with conventional distance-, fading-, and coverage-based association schemes.

\section*{Appendix A \\ Proof of Corollary 1}

\begin{proof}
$
\psi_{l,k}
=
\frac{\mathcal D_{l,k}}{\mathcal I_{l,k}},
\psi_{m,k}
=
\frac{\mathcal D_{m,k}}{\mathcal I_{m,k}}.
$
Thus,
\begin{equation}
\psi_{l,k}>\psi_{m,k}
\end{equation}
is equivalent to
\begin{equation}
\frac{\mathcal D_{l,k}}{\mathcal I_{l,k}}
>
\frac{\mathcal D_{m,k}}{\mathcal I_{m,k}}.
\end{equation}
Since
$
\mathcal I_{l,k}>0,
$
$
\mathcal I_{m,k}>0,
$
we have
\begin{equation}
\mathcal D_{l,k}\mathcal I_{m,k}
>
\mathcal D_{m,k}\mathcal I_{l,k}.
\end{equation}
Even if
$
\mathcal D_{l,k}<\mathcal D_{m,k},
$
the inequality
$
\mathcal D_{l,k}\mathcal I_{m,k}
>
\mathcal D_{m,k}\mathcal I_{l,k}
$
can still hold when
$
\frac{\mathcal I_{l,k}}{\mathcal I_{m,k}}
<
\frac{\mathcal D_{l,k}}{\mathcal D_{m,k}}.
$
Therefore, AP $l$ can be preferred over AP $m$ even when its desired effective signal power is smaller, provided that its effective disturbance power is sufficiently smaller. 
This completes the proof.
\end{proof}

\section*{Appendix B \\ Proof of Corollary 2}
\begin{proof}
From the definition of $f(z_{lk})$,
$
f(z_{lk})
=
z_{lk}
\left(
r^2+z_{lk}^2
\right)^{-\frac{\alpha+1}{2}}.
$
Then
\begin{equation}
\begin{aligned}
\frac{\partial f}{\partial z_{lk}}
&=
\left(
r^2+z_{lk}^2
\right)^{-\frac{\alpha+1}{2}}
+
z_{lk}
\left[
-\frac{\alpha+1}{2}
\left(
r^2+z_{lk}^2
\right)^{-\frac{\alpha+3}{2}}
2z_{lk}
\right]
\\
&=
\left(
r^2+z_{lk}^2
\right)^{-\frac{\alpha+1}{2}}
-
(\alpha+1)z_{lk}^2
\left(
r^2+z_{lk}^2
\right)^{-\frac{\alpha+3}{2}}
\\
&=
\left(
r^2+z_{lk}^2
\right)^{-\frac{\alpha+3}{2}}
\left[
r^2+z_{lk}^2-(\alpha+1)z_{lk}^2
\right]
\\
&=
\left(
r^2+z_{lk}^2
\right)^{-\frac{\alpha+3}{2}}
\left(
r^2-\alpha z_{lk}^2
\right).
\end{aligned}
\end{equation}
Since
$
\left(
r^2+z_{lk}^2
\right)^{-\frac{\alpha+3}{2}}
>
0,
$
we have
$
\frac{\partial f}{\partial z_{lk}}>0
\iff
r^2-\alpha z_{lk}^2>0
\iff
0<z_{lk}<\frac{r}{\sqrt{\alpha}},
$
$
\frac{\partial f}{\partial z_{lk}}=0
\iff
z_{lk}=\frac{r}{\sqrt{\alpha}},
$
and
$
\frac{\partial f}{\partial z_{lk}}<0
\iff
z_{lk}>\frac{r}{\sqrt{\alpha}}.
$
Therefore,
$
z_{lk}^\star
=
\frac{r}{\sqrt{\alpha}}
$
is the unique maximizer of $f(z_{lk})$.
\end{proof}

\section*{Appendix C \\ Proof of Corollary 3}

\begin{proof}
From \eqref{eq:deltaR_revised},
$
\Delta R_k(J)
=
\log_2
\left(
1+
\frac{
\gamma_{(J)k}
}{
1+
\sum_{j=1}^{J-1}
\gamma_{(j)k}
}
\right).
$
If
\begin{equation}
\gamma_{(J)k}
\leq
\varepsilon
\left(
1+
\sum_{j=1}^{J-1}
\gamma_{(j)k}
\right),
\end{equation}
then
\begin{equation}
\frac{
\gamma_{(J)k}
}{
1+
\sum_{j=1}^{J-1}
\gamma_{(j)k}
}
\leq
\varepsilon.
\end{equation}
Therefore,
\begin{equation}
1+
\frac{
\gamma_{(J)k}
}{
1+
\sum_{j=1}^{J-1}
\gamma_{(j)k}
}
\leq
1+\varepsilon.
\end{equation}
Since $\log_2(\cdot)$ is monotonically increasing,
$
\Delta R_k(J)
\leq
\log_2(1+\varepsilon).
$
This completes the proof.
\end{proof}

\bibliographystyle{Bibliography/IEEEtranTIE}
\bibliography{Bibliography/IEEEabrv,Bibliography/mybibfile.bib} 

\end{document}